\documentclass[11pt]{article}
\usepackage{style}

\newtheorem{theorem}{Theorem}[section]

\usepackage{aliascnt}
\newaliascnt{definition}{theorem}
\newtheorem{definition}[definition]{Definition}
\aliascntresetthe{definition}

\newaliascnt{lemma}{theorem}
\newtheorem{lemma}[lemma]{Lemma}
\aliascntresetthe{lemma}

\newaliascnt{proposition}{theorem}
\newtheorem{proposition}[proposition]{Proposition}
\aliascntresetthe{proposition}

\newaliascnt{conjecture}{theorem}
\newtheorem{conjecture}[conjecture]{Conjecture}
\aliascntresetthe{conjecture}

\newtheorem{remark}{Remark}

\def\authorNote{} % comment this line to remove the author note

\ifdefined\authorNote
\newcommand{\yc}[1]{{\color{magenta} [{YC:} #1]}}
\newcommand{\nai}[1]{{\color{brown} [{Nai:} #1]}}
\newcommand{\atsuya}[1]{{\color{OliveGreen} [{Atsuya:} #1]}}
\newcommand{\francois}[1]{{\color{blue} [{Francois:} #1]}}

\else
\newcommand{\yc}[1]{}%{{\color{blue} [{Someone:} #1]}}
\newcommand{\nai}[1]{}%{{\color{brown} [{Nai:} #1]}}
\newcommand{\atsuya}[1]{}
\newcommand{\francois}[1]{}
\fi

\title{Fine-Grained Complexity for Quantum Problems from Size-Preserving Circuit-to-Hamiltonian Constructions}

\author[1]{Nai-Hui Chia\thanks{nc67@rice.edu}}
\author[2]{Atsuya Hasegawa\thanks{atsuya.hasegawa@math.nagoya-u.ac.jp}}
\author[2]{Fran{\c{c}}ois Le Gall\thanks{legall@math.nagoya-u.ac.jp}}
\author[1]{Yu-Ching Shen\thanks{ycshen@rice.edu}}

\affil[1]{\textit{Department of Computer Science, Rice University, USA}}
\affil[2]{\textit{Graduate School of Mathematics, Nagoya University, Japan}}

\date{}

\begin{document}

\maketitle

\begin{abstract}
    The local Hamiltonian (LH) problem is the canonical $\mathsf{QMA}$-complete problem introduced by Kitaev. In this paper, we show its hardness in a very strong sense: we show that the 3-local Hamiltonian problem on $n$ qubits cannot be solved classically in time $O(2^{(1-\varepsilon)n})$ for any $\varepsilon>0$ under the Strong Exponential-Time Hypothesis (SETH), and cannot be solved quantumly in time $O(2^{(1-\varepsilon)n/2})$ for any $\varepsilon>0$ under the Quantum Strong Exponential-Time Hypothesis (QSETH). These lower bounds give evidence that the currently known classical and quantum algorithms for LH cannot be significantly improved.

    Furthermore, we are able to demonstrate fine-grained complexity lower bounds for approximating the quantum partition function (QPF) with an arbitrary constant relative error. Approximating QPF with relative error is known to be equivalent to approximately counting the dimension of the solution subspace of $\mathsf{QMA}$ problems. We show the SETH and QSETH hardness to estimate QPF with constant relative error. We then provide a quantum algorithm that runs in $O(\sqrt{2^n})$ time for an arbitrary $1/\mathrm{poly}(n)$ relative error, matching our lower bounds and improving the state-of-the-art algorithm by Bravyi, Chowdhury, Gosset, and Wocjan (Nature Physics 2022) in the low-temperature regime.
	
    To prove our fine-grained lower bounds, we introduce the first size-preserving circuit-to-Hamiltonian construction that encodes the computation of a $T$-time quantum circuit acting on $N$ qubits into a $(d+1)$-local Hamiltonian acting on $N+O(T^{1/d})$ qubits. This improves the standard construction based on the unary clock, which uses $N+O(T)$ qubits.
\end{abstract}

\clearpage

\tableofcontents

 \clearpage

\section{Introduction}

\subsection{Background}

\paragraph{Local Hamiltonian (LH) problem.}
Kitaev \cite{Kitaev+02} introduced the $k$-local Hamiltonian ($\kLH$) problem, which asks to estimate the ground-state energy of a $k$-local Hamiltonian 
\[
	H=\sum_{i=1}^m H_i\,,
\]
acting on $n$ qubits,
where $m=\poly(n)$ and $\norm{H}\le \poly(n)$, with constant additive error.\footnote{Equivalently, we can define the local Hamiltonian problem as $||H|| \leq 1$ and $1/\poly(n)$ additive error. Our definition is the same as \cite{Kempe+06}.} Kitaev, at the same time, showed that the problem is $\QMA$-complete, i.e., any efficient quantum verification procedure can be embedded into the local Hamiltonian problem. The embedding procedure is called the circuit-to-Hamiltonian construction, which is a quantum analog of the Cook-Levin theorem \cite{Cook1971,Levin1973}. The local Hamiltonian problem is in the interface between physics and computer science, and has been the central problem in quantum complexity theory. Subsequent work showed that, when we reduce the locality to 2, restrict the interaction of local terms to geometrically local, and consider some specific type of Hamiltonians, the problem is still $\QMA$-hard \cite{Biamonte+08,Cubitt+16,Kempe+06,Oliveira+08,Piddock+17}.

The $k$-local Hamiltonian problem can be solved classically in $\myO{2^n}$ time using the power method or its variant the Lanczos method \cite{KW92,Lanczos50}, for $k=O(\log n)$.\footnote{In this paper the notation $\myO{\cdot}$ suppresses $\poly(n)$ factors.} 
On a quantum computer, we can use Grover search \cite{GroverSTOC96} combined with other techniques to achieve running time $\myO{2^{n/2}}$ \cite{Apeldoorn+20,Ge+19,Gilyen+STOC19,Kerzner+24,Lin+20,Martyn+21,Poulin+09}.\footnote{We note that some previous works showed how to obtain larger (``super-Grover'') speedups for certain applications in combinatorial optimization \cite{Hastings18,Dalzell+STOC23} or for computing a coarser approximation of the ground-state energy \cite{Buhrman+PRL25}.}

\paragraph{Approximating quantum partition function with relative error.} Computing the quantum partition function is an even harder problem. 
The partition function is a function of the temperature and the Hamiltonian of the system. 
It is defined by
\[
    Z:=\tr[e^{-\beta H}],
\]
where $\beta$ is the inverse of the temperature and $H$ is the Hamiltonian.
The \emph{Quantum Partition Function (QPF)} problem asks us to compute the value of $Z$ for a given Hamiltonian at a specified temperature with relative error.\footnote{It is crucial here to consider relative error estimation. When we consider the problem to estimate the (normalized) QPF with additive error $1/\poly(n)$, it is known to be $\mathsf{DQC1}$-complete \cite{Brandao08,Chowdhury+21} where $\mathsf{DQC1}$ is the set of problems that can be solved efficiently with only ``one-clean qubit'' \cite{Knill+98}.}
Intuitively, QPF is harder than LH because evaluating the partition function requires information about the \emph{entire} spectrum of the Hamiltonian, whereas solving LH only involves the ground-state energy. 
In fact, we can show that LH reduces to approximating the QPF, even with a constant relative error, and thus approximating QPF with constant-relative error is $\QMA$-hard.\footnote{To the best of our knowledge, \cite{Bravyi+22} first mentioned that approximating QPF with a relative error is QMA-hard; we provide a self-contained formal proof in \Cref{subsec:seth-hardness of QPF}.}
Furthermore, Bravyi et~al.~\cite{Bravyi+22} show that approximating the QPF up to constant-relative error is computationally equivalent to approximately counting the number of witnesses accepted by a $\QMA$ verifier.

For algorithms, Poulin and Wocjan \cite{PW09b}
demonstrated a quantum algorithm that approximates the partition function $Z$ of an $n$-qubit system up to relative error $\delta$ in $O^\ast\left(\frac{1}{\delta}\sqrt{\frac{2^n}{Z}}\right)$ time. Later Bravyi, Chowdhury, Gosset, and Wocjan ~\cite{Bravyi+22} proposed another algorithm achieving the same running time while improving the space complexity.

When the temperature is not too low (i.e., $\beta$ is not too large), $Z$ is not exponentially small (for example, $Z\approx O(1)$ when $\beta\approx0.7n$) and the running time is then roughly $O^\ast(2^{n/2})$. \vspace{-2mm}

\paragraph{Our main question.}
In summary, existing classical and quantum algorithms for LH run in $O^*(2^{n})$ and $O^*(2^{n/2})$ time, respectively, suggesting a fundamental barrier to further improvements. The situation for QPF is even more mysterious: no $O(2^n)$-time classical algorithm is currently known to the best of our knowledge, and existing quantum algorithms fail to achieve the expected quadratic speedup in the low-temperature regime. These observations naturally motivate the following questions:\vspace{-2mm}

\begin{center}
\emph{Do there exist classical or quantum algorithms that are significantly faster than $O(2^{n})$ or $O(2^{n/2})$ for computing the ground-state energy? What is the best possible algorithms for approximating quantum partition function even in the low-temperature regime? }
\end{center}

\paragraph{Fine-grained complexity, SETH and QSETH.}
Our goal is to provide a $\Omega(2^{n})$ classical (resp. $\Omega(2^{n/2})$ quantum) lower bound.
We are going to use the language of fine-grained complexity \cite{VassilevskaWilliams15,VassilevskaWilliams18,Nederlof26}, which is a relatively recent field of computational complexity theory that investigates complexity in a more fine-grained way than the conventional class like $\mathsf{P}$ and $\NP$. 

The $k$-local Hamiltonian problem is a generalization of $\kSAT$, the central problem in classical complexity theory and its fine-grained complexity has been well studied. 
The Strong Exponential Time Hypothesis (SETH) \cite{Impagliazzo+01} states that $\kSAT$ needs roughly $2^n$ time for large~$k$ (see \cref{conjecture:SETH}).\footnote{We stress that (Q)SETH are about the complexity of $\kSAT$ for large $k$. For small values of~$k$, an algorithms better than $O(2^{n/2})$ are known. For example, there exists a $1.32793^n$-time classical algorithm for $3\textrm{-}\mathrm{SAT}$ \cite{LiuICALP18}.} More recently, a quantum analog was proposed in \cite{Aaronson+CCC20,Buhrman+21} as Quantum Strong Exponential-Time Hypothesis (QSETH). The conjecture is to claim that any quantum algorithm to solve $\kSAT$ needs $\Omega(2^{n/2})$ time (see \cref{conjecture:QSETH}).
SETH and QSETH are regarded as plausible computational conjectures and serve as a useful tool for deriving conditional lower bounds for various fundamental problems. 
First, the SAT problem has been studied for decades, and despite this extensive effort, the $\Omega(2^{n})$ (resp. $\Omega(2^{n/2})$) lower bound for classical (resp. quantum) algorithms has not been broken; designing quantum algorithms that refute these conjectures would likely require fundamentally new techniques and could lead to major breakthroughs in algorithm design. 
In addition, (Q)SETH has been used to establish conditional lower bounds for a wide range of problems, such as the orthogonal vectors problem~\cite{BRSV17, Aaronson+CCC20}, the closest pair problem~\cite{DSL19, KM20,  Aaronson+CCC20,Buhrman+21}, the edit distance problem~\cite{BI15, Buhrman+21}, lattice problems~\cite{CCK+23,HKW24}. These results highlight the value of (Q)SETH as a powerful framework for studying the hardness of computational problems.\vspace{-2mm}

\paragraph{Barriers to fine-grained complexity lower bounds for $\boldsymbol{O(1)}$-local Hamiltonians.}

We aim to prove the fine-grained complexity lower bound for LH and QPF assuming SETH and QSETH. One straightforward approach is to map the clauses in the $\kSAT$ instance $\Phi$ to projectors that act on local qubits. The sum of these local projectors is a $k$-local Hamiltonian $H$ whose ground subspace encodes the solutions for $\Phi$ 
(see \cref{sec:appexdixA} for the construction). In this reduction, the number of qubits of $H$ is equal to the number of variables in $\Phi$ and therefore directly gives a $\Omega({2^n})$ (resp. $\Omega(2^{n/2})$) lower bound assuming SETH (resp. QSETH).
However, the locality of the induced Hamiltonian is precisely $k$, where $k$ is the parameter of the $\kSAT$ hard instance; this value depends on $\varepsilon$, as characterized by SETH and QSETH. Hence, this approach fails to establish  $\Omega({2^n})$ (resp. $\Omega(2^{n/2})$) lower bounds for Hamiltonians with fixed locality (e.g., 3-local).

Another approach is to use known circuit-to-Hamiltonian reductions,
in which a verification circuit $U=U_T\cdots U_1$ for a $\QMA$ problem (e.g., a verification circuit of size $T=\poly(n)$  for a $\kSAT$ instance) is embedded into a local Hamiltonian $H$.
In particular, for any verification circuit, there are reductions to make the induced Hamiltonian 5-local \cite{Kitaev+02} or 2-local \cite{Kempe+06}, which achieves our purpose of fixed locality. Following the idea by Feynman \cite{Fey86}, the standard circuit-to-Hamiltonian construction \cite{Kitaev+02,Kempe+06} introduces a clock register, which is used to represent $T+1$ orthogonal clock states.
The ground state of the resulting Hamiltonian is the history state 
\[
\frac{1}{\sqrt{T+1}} \sum_{i=0}^T U_i\cdots U_0 \ket{\psi}_\textrm{circuit} \otimes \ket{i}_\textrm{clock},
\]
where $\ket{\psi}_\textrm{circuit}$ denotes the initial state of the circuit (including the witness) and $\ket{i}_\textrm{clock}$ denotes the clock states. However, this standard reduction results in significant size overhead due to the number of additional qubits needed to implement the clock states: 
in the standard reduction, a unary clock \cite{Kitaev+02,Kempe+06} is used, which requires $O(T)$ qubits. Note that the analysis to obtain fine-grained lower bounds is very sensitive to the size overhead between the circuit and the local Hamiltonian. In particular, in order to obtain the desired lower bounds based on (Q)SETH, the clock need to be implementable in $o(n)$ qubits.
Although we could use the binary clock \cite{Kitaev+02} to reduce the size of the clock register from $O(T)$ qubits to $O(\log T)$ qubits, the locality would then increase to $O(\log T)$.

In summary, we require a $\kSAT$-to-LH reduction that simultaneously satisfies the following two conditions---a combination that existing methods fail to achieve:

\begin{enumerate}
    \item[(1)] \emph{Size-preserving}: the resulting Hamiltonian acts on $n + o(n)$ qubits, and
    \item[(2)] \emph{Locality-independent}: its locality is a constant independent of~$k$.
\end{enumerate}

\subsection{Our results}

In this paper, to investigate the fine-grained complexity for $\QMA$-hard problems, we show the first size-preserving circuit-to-Hamiltonian construction. We denote by $\lambda(H)$ the smallest eigenvalue (ground-state energy) of a local Hamiltonian $H$.

\begin{theorem}[Informal version of \cref{theorem:main}]\label{th:main-inf}
    Let $L = (L_\mathrm{yes},L_\mathrm{no})$ be a promise language in $\QMA$, and $x \in \{0,1\}^n$ be an input for the language $L$. Let $U_x$ be a verification circuit acting on $N$ qubits and $T$ be the number of elementary gates of $U_x$. Then, for any integer $d\geq1$, there exists a $(d+1)$-local Hamiltonian $H_x$ acting on $N + O(T^{1/d})$ qubits satisfying the following properties;
    \begin{itemize}
        \item $||H_x|| \leq \poly(n)$,
        \item if $x \in L_\mathrm{yes}$, $\lambda(H_x) \leq E_\mathrm{yes}$,
        \item if $x \in L_\mathrm{no}$, $\lambda(H_x) \geq E_\mathrm{no}$,
        \item $E_\mathrm{no} - E_\mathrm{yes} \geq \Omega(1)$.
    \end{itemize}
\end{theorem}

Our construction is the first construction to achieve both the $O(1)$-locality and the $o(T)$-qubit size for the clock register. See \cref{tab:size and locality tradeoff} for a comparison between our construction and the standard constructions.

\begin{table}[h]
    \centering
    \begin{tabular}{c|c|cc}
         Clock & Locality & Size & \\ \hline \hline
         Binary \cite{Kitaev+02} & $O(\log T) $ & $N+O(\log T)$ & \\ \hline
         Unary \cite{Kitaev+02,Kempe+06} & $2$ &  $N+O(T)$ & \\ \hline
         Our work (\cref{theorem:main}) & $d+1$ & $N+O(T^{1/d})$ & \\
    \end{tabular}
    \caption{Comparison between our result and standard circuit-to-Hamiltonian constructions. Locality and Size indicate the locality and size of the local Hamiltonian reduced from a quantum circuit of $T$ gates. $N$ denotes the qubit size of verification quantum circuits. Our result holds for any integer $d\geq1$.}
    \label{tab:size and locality tradeoff}
\end{table}

We then show $\kSAT$ has a very efficient verification procedure as a unitary circuit (\cref{lem:U_compute_Phi}). By combining this with our size-preserving circuit-to-Hamiltonian construction, we obtain the SETH- and QSETH-hardness for the local Hamiltonian problem.

\begin{theorem}[Informal version of \cref{thm:SETH hardness for 3LH}]\label{th:main-classical}
	Assuming SETH (resp. QSETH), estimating the ground-state energy of a $3$-local Hamiltonian problem with any constant additive error cannot be solved by any classical algorithm in time $\myO{2^{(1-\varepsilon)n}}$ (resp. by any quantum algorithm in time $\myO{2^{(1-\varepsilon)n/2}}$) for any constant $\varepsilon>0$. 
\end{theorem}

Our lower bound for the local Hamiltonian problem matches the performance of the best known classical and quantum algorithms.

\begin{remark}
Recently, Buhrman et al. \cite{Buhrman+PRL25} constructed classical and quantum algorithms running in time $2^{n(1-\Omega(\delta/M))}$ and $2^{n(1-\Omega(\delta/M))/2}$, respectively, where $M:=\sum_{i}\|H_i\|$, to estimate the ground-state energy of a $O(1)$-local Hamiltonian up to an additive error $\delta$. These upper bounds do not contradict our conditional lower bounds since \cref{th:main-classical} is about estimating the ground-state energy within constant additive error, and the upper bounds given by \cite{Buhrman+PRL25} become $O(2^{n(1-1/\poly(n))})$ and $O(2^{n(1-1/\poly(n))/2})$ in this regime.
\end{remark}

We further investigate fine-grained complexity of quantum partition functions. We first show the SETH- and QSETH-hardness to estimate QPF with constant relative error.

\begin{theorem}
    [Informal version of \Cref{thm:main_qpf}]
    \label{thm:main_qpf_informal}
    Assuming SETH (resp. QSETH), approximating the quantum partition function for 3-local Hamiltonians to within any constant relative error cannot be done by any classical algorithm in time $O\bigl(2^{n(1-\varepsilon)}\bigr)$ (resp. quantum algorithm in time $O\bigl(2^{n(1-\varepsilon)/2}\bigr)$) for any constant $\varepsilon > 0$. 
\end{theorem}

We then show a new algorithm for QPF.

\begin{theorem}[Informal version of \cref{thm:main_qpf_alg}]
    There exists a quantum algorithm whose running time is $\myO{2^{n/2}}$ to approximate the quantum partition function for a constant-local Hamiltonian on $n$ qubits up to a relative error ${1}/{\poly(n)}$.
\end{theorem}

We now compare our algorithm with the one proposed in~\cite{PW09b} and~\cite{Bravyi+22}.
The running time of their algorithm is $O^\ast\left(\frac{1}{\delta}\sqrt{\frac{2^n}{Z}}\right)$,
which depends on the value of the partition function $Z$. 
When the inverse temperature $\beta = O(n^c)$ for some large constant $c$ (i.e., at low temperatures), the running time may exceed $O(2^{n/2})$. 
In contrast, the running time of our algorithm does not depend on $Z$: it runs in $O^\ast(2^{n/2})$ for all $\beta = \poly(n)$.

\subsection{Our techniques}

\paragraph{Compressed clock state and how to use it for the circuit-to-Hamiltonian construction.}

To ensure that our reduction works, we need a construction of the clock state associated with the induced Hamiltonian that satisfies the following three conditions: 
\begin{enumerate}
    \item[(i)] The clock state uses $o(n)$ of qubits and can encode $\poly(n)$ of computation steps,
    \item[(ii)] the operation on the clock state needs to be constant-local, with locality is independent of the circuit, and
    \item[(iii)] the computation process of $U$ is encoded in the ground state of the induced Hamiltonian $H$.  
\end{enumerate}

The unary clock satisfies conditions (ii) and (iii), but does not satisfy condition (i).
In~\cite{Chia+CCC23}, a clock construction is proposed that satisfies conditions (i) and (ii). To represent $O(a^d)$ time on an $a$-qubit register for integers $a \geq d \geq 1$, they use a Johnson graph where the vertices of the Johnson graph $(a,d)$ are the $d$-element subset of an $a$-element set and two vertices are adjacent when the intersection of the two vertices (subsets) contains $(d-1)$-elements (see \cref{def:Yao-Ting clock} for a formal definition\footnote{Not to confuse readers, in the introduction we denote by $a$ the register size of the $(a,d)$-clock state and by $n$ the input size.}).
However, unlike the circuit-to-Hamiltonian constructions \cite{Kitaev+02,Kempe+06}, the circuit-to-Hamiltonian reduction in~\cite{Chia+CCC23} encodes the computation in the \emph{time evolution} of $H$, rather than in its ground state. As a result, there is no guarantee regarding the structure or properties of the ground state in their construction.

We build on the clock-state construction from~\cite{Chia+CCC23} and a 2-local Hamiltonian construction from Section 5 in \cite{Kempe+06}.\footnote{Section 6 of \cite{Kempe+06} introduces another approach to obtain a 2-local Hamiltonian using the perturbation theory.}
By introducing a carefully designed penalty term into $H$, we enforce that the computation history is stored in the ground state. 
This yields a novel circuit-to-Hamiltonian reduction that simultaneously preserves the size of the Hamiltonian, maintains constant locality, and ensures that the circuit $U$ is encoded in the ground state of the resulting Hamiltonian $H$. 
Our construction satisfies all three requirements (i), (ii) and~(iii), overcoming the limitations that neither of the existing reductions can address individually.

\paragraph{Additional techniques to reduce the locality.}

From the approach above, we have a $d+2$-local Hamiltonian from a quantum circuit whose number of gates is $O(a^d)$. This is because we need a local term operating on $d+1$ qubits in the clock register and 1 qubit in the circuit register for the propagation Hamiltonian.\footnote{We do not need to consider 2-local gates thanks to the technique in Section 5 of \cite{Kempe+06}.}

To reduce the locality further, we introduce a new clock construction that combines the $(a,d-1)$-clock state from \cite{Chia+CCC23} and the unary clock on $a$ qubits. Crucial properties of the two clocks are:
        \begin{itemize}
	       \item for the $(a,d-1)$ clock, time can be identified by looking at $(d-1)$ position ($(d-1)$-locality), while increasing or decreasing time by one unit can be done by modifying $d$ positions ($d$-locality);
	       \item for the unary clock, increasing or decreasing time by one unit can be done by modifying one position (1-locality). Moreover, if time is $0$ or $a$, it can be identified by looking at $1$ position ($1$-locality) (see \cref{tab:operators in KKR06} in \cref{subsec:KKR}).
        \end{itemize}
By making the unary clock move periodically from $0$ to $a$ and from $a$ to $0$, we increase and decrease most of the clock time by $d$-local operators. By carefully analyzing the locality of all local Hamiltonian, we have a $(d+1)$-local Hamiltonian from a circuit of $O(a^d)$ gates using the $2a$-qubit clock register.

The remaining task is to construct a verification circuit that uses $o(n)$ ancilla qubits for any $k$-SAT instance whose number of variables is $n$. 
To achieve this, we use a counter to record the number of clauses satisfied by a given assignment. 
The $k$-SAT formula is satisfied if and only if the value stored in the counter equals $m$. 
Since counting up to $m$ requires only $\log m$ qubits, the number of ancilla qubits needed is $n_a = O(\log m) = o(n)$, as required. The number of gates for the verification circuit is $O(n  \log^2 n)$ (\cref{lem:U_compute_Phi}) and thus we can use $d=2$ and obtain the SETH- and QSETH-hardness for the 3-local Hamiltonian problem.

\paragraph{Fine-grained complexity and quantum algorithms for the QPF problem.}

The lower bound for QPF follows directly from the lower bound for LH. When $\beta$ is a sufficiently large polynomial, the partition function is dominated by low-energy states. If the ground-state energy is at most $\EnergyYes$, then by ignoring all other states except the ground state, we obtain $Z \ge e^{-\beta \EnergyYes}$. Conversely, if the ground-state energy is at least $\EnergyNo$, then by treating all eigenstates as having energy $\EnergyNo$, we get $Z \le 2^n e^{-\beta \EnergyNo}$. Therefore, by approximating $Z$, we can distinguish between the two cases and decide the ground-state energy. To ensure this distinction, it suffices that $(1 - \delta) e^{-\beta \EnergyYes} > (1 + \delta) 2^n e^{-\beta \EnergyNo}$. In our reduction, we have observed that for any $\delta > 0$, this inequality holds for all sufficiently large~$n$.

To obtain an $O(\sqrt{2^n})$-time quantum algorithm for the QPF with $1/\poly(n)$ relative error, we adopt a technique from the proof of the equivalence between approximating the number of witnesses for a QMA verifier and approximating the quantum partition function~\cite{Bravyi+22}. Briefly, the approach involves partitioning the energy spectrum into polynomially many intervals and selecting a representative energy value $\mathcal{E}_j$ for each interval $j$. By estimating the number of eigenstates $M_j$ within each interval, the partition function $Z$ can be approximated as:
\[
\tilde{Z} = \sum_j M_j e^{-\beta \mathcal{E}_j}.
\]

We now explain how to count $M_j$. 
First, we construct an energy estimation circuit $U_{EE}$ that outputs the corresponding eigenvalue $E_p$ for each eigenstate $\ket{\psi_p}$ of $H$. 
In other words, $U_{EE}\ket{\psi_p} = \ket{E_p}\ket{\psi_p}$. 
The circuit $U_{EE}$ can be implemented using Hamiltonian simulation together with phase estimation. Then, we apply a circuit $U_{dec}$ that decides whether $E_p$ is in the interval $j$.
Let $U_j:=U_{dec}U_{EE}$. Together, we have $U_j\ket{\psi_p}=\ket{1}\ket{\phi_p}$ if $E_p$ in the interval $j$ and $U_j\ket{\psi_p}=\ket{0}\ket{\phi_p}$ otherwise, where $\ket{\phi_p}$ is some quantum state corresponding to $\ket{\psi_p}$. 

Now, let us suppose that the following two ``too-good-to-be-true'' assumptions hold: 
\begin{enumerate}
    \item A uniform superposition of all eigenstates can be efficiently prepared, and
    \item the energy estimation circuit, $U_{EE}$, can compute eigenvalues with zero error.
\end{enumerate}
Then, by applying $U_j$ on the state, we obtain
\begin{equation}
    \label{eq:tech_count}
    U_j\sum_{p}\frac{1}{\sqrt{N}}\ket{\psi_p} = \sqrt{\frac{M_j}{N}}\ket{1}\ket{\xi_1} + \sqrt{\frac{N-M_j}{N}}\ket{0}\ket{\xi_0},
\end{equation}
where $N=2^n$ is the number of the eigenstates and $\ket{\xi_1}$, $\ket{\xi_0}$ are two quantum states orthogonal to each other. Finally, by using the well-known amplitude estimation algorithm on \Cref{eq:tech_count}, we can obtain $M_j$. This algorithms runs in $O(\sqrt{2^n})$ time. 

We then remove the two ``too-good-to-be-true'' assumptions while keeping the $O(\sqrt{2^n})$ runtime. First, directly preparing the uniform superposition of all eigenstates is generally computationally hard since the eigenstates are unknown. To overcome this difficulty, a commonly used technique is to use EPR state (e.g., \cite{PW09b, Kerzner+24}). We have that uniform superposition over a complete basis, tensored with its complex conjugate, is equivalent to an EPR state.
Therefore, by applying $U_j$ to the EPR state, we obtain
\begin{equation}
    \label{eq:tech_count_epr}
    U_j \sum_{p} \frac{1}{\sqrt{N}}\ket{\psi_p}\ket{\psi_p^{\ast}}
    = \sqrt{\frac{M_j}{N}}\ket{1}\ket{\xi'_1} 
    + \sqrt{\frac{N - M_j}{N}}\ket{0}\ket{\xi'_0},
\end{equation}
where $\ket{\xi'_1}$ and $\ket{\xi'_0}$ are orthogonal quantum states. 
By applying the amplitude estimation algorithm to the state in~\Cref{eq:tech_count_epr}, we can approximate $M_j$.

Second, the quantum energy estimation circuit described above introduces an additive error of $1/\poly(n)$, which can significantly impact the approximation of the quantum partition function. Specifically, if the eigenstates are densely concentrated near the boundaries of the intervals in the selected partition (within the precision of the energy estimation), we may ``doubly count'' these states. Such inaccuracies in state counting translate directly into a constant relative error in the final estimation of the partition function.

To reduce the constant relative error to $1/\poly(n)$, we observe that the ``double-counting'' issue only arises when the chosen partition unluckily has significant fraction of eigenvalues lie near the boundaries of intervals. Intuitively, by considering polynomially many partitions with mutually disjoint boundaries, any specific eigenstate can be miscounted in at most a constant number of these partitions. Consequently, if we run the algorithm over all such partitions, an averaging argument ensures that at least one set will have a number of double-counted eigenstates that is at most an inverse polynomial fraction of the total. This approach yields an algorithm achieving $1/\poly(n)$ relative error in $O(\sqrt{2^n})$ time.

\subsection{Open questions}

\begin{enumerate}
    \item It is known that the 2-local Hamiltonian problem is $\QMA$-hard \cite{Kempe+06}. 
    \textbf{Can we show a size-preserving reduction to 2-local Hamiltonian and its (Q)SETH-hardness?} In our clock construction, we need to touch two clocks and simultaneously apply a 1-qubit gate, so we cannot improve the locality directly. Note that perturbation theory does not seem to work here either: the standard perturbation technique adds an ancilla qubit for each term, and thus does not give a size-preserving construction
    when applied directly to our 3-local Hamiltonian.
    
    \item  Even when we consider geometrically local \cite{Oliveira+08} or specific $2$-local terms such as Heisenberg interaction, antiferromagnetic $XY$ interactions \cite{Cubitt+16, Piddock+17}, $2$-local $ZX$ and $ZZXX$ Hamiltonian \cite{Biamonte+08}, the problem is still $\QMA$-hard.
    This leads to a natural question: \textbf{Does the (Q)SETH-hardness still hold for geometrically local Hamiltonians whose each local term is restricted to a specific form?}
    
    \item Our relative threshold gap $(\EnergyNo-\EnergyYes)/M$ in the $3$-local Hamiltonian problem is $1/p(n)$ where $M:=\sum_{i=i}^{m}\|H_i\|$ and $p(n)$ is a polynomial with large degree. 
    A natural question arises: \textbf{when the relative gap increases, for example, $\boldsymbol{(\EnergyNo-\EnergyYes)/M = \Omega(1)}$, what conditional lower bound can we have under plausible conjectures?} 
    Notably, the algorithm of~\cite{Buhrman+PRL25} runs in time $O^{\ast}(2^{\frac{n}{2}(1-\Omega(1))})$ when $(\EnergyNo-\EnergyYes)/M = \Omega(1)$, and therefore having the same lower bound $\Omega(2^{\frac{n}{2}(1-o(1))})$ for the $(\EnergyNo-\EnergyYes)/M=\Omega(1)$ regime as we have for $(\EnergyNo-\EnergyYes)/M=1/p(n)$ in our work is impossible.
    \item In this paper, we show the SETH-hardness of the 3-local Hamiltonian problem, which is a quantum analog of MAX-3-SAT, and it is known that there exists a faster exponential time algorithm for MAX-2-SAT, say $O^*(1.732^n)$-time algorithm by \cite{WilliamsTCS05}. \textbf{Can we show the fine-grained complexity of MAX-3-SAT?}
    
    \item  In this work, we show that $\LH$ can be reduced to $\QPF$, and the reduction is fine-grained even for an arbitrary constant relative error; therefore, we can obtain the $\Omega(2^{n/2})$ lower bound assuming (Q)SETH. Notably, the lower bound holds for any constant relative error and locality at least $3$. Along this line, we can ask \textbf{whether the lower bound continues to hold when considering $2$-local Hamiltonian, geometrical restrictions, and specified local terms, or we can design faster algorithms under these constraints.}
    
    \item Another interesting question is \textbf{whether the reverse reduction holds, namely, whether approximating QPF can be reduced to the LH problem.} At first glance, this seems unlikely, as QPF can be interpreted as estimating the dimension of the ground space in the low-temperature regime. However, since we are considering exponential lower bounds, such a reduction could, in principle, be allowed to run in super-polynomial time.
    
    \item Our lower bound for the 3QPF problem holds for $\beta = O(n)$ and $\|H\|$ being a large degree polynomial. A natural question to ask is \textbf{what happens as the temperature increases.} In other words, when $\beta\|H\|$ becomes small, does the lower bound still hold in this regime? Or, at what point do we observe a \emph{phase transition}?
\end{enumerate}

\subsection{Organization}

In \cref{sec:preliminary}, we give preliminaries for the rest of the paper. In \cref{sec:size-preserving circuit-to-Hamiltonian construction}, we show the size-preserving circuit-to-Hamiltonian construction by introducing a new clock. In \cref{sec:SETH-hardness of LH}, we show the SETH-hardness of the local Hamiltonian problem. In \cref{sec:fine-grained complexity of QPF new}, we show the fine-grained complexity to estimate the quantum partition function with relative error.

\paragraph{Previous version.}
A preliminary version of this work was submitted to ArXiv in October 2025 under a different title \cite{Chia+25}. The present version corrects an issue in one of the proofs, strengthens both the results and revises the presentation, placing greater emphasis on the size-preserving circuit-to-Hamiltonian construction and lower bounds based on the (classical) SETH.
% \newpage
\section{Preliminaries}\label{sec:preliminary}

We refer to \cite{NC10,deWolf19} for references in quantum computing, \cite{Gharibian2015} for a reference and survey in quantum Hamiltonian complexity.

\subsection{General notation}
For $n\in\mathbb{N}$, we use $[n]$ to denote the set collecting all the positive integers smaller than or equal to $n$, that is, $\{1,2,\dots,n\}$. 

For a bit string $x\in\{0,1\}^{*}$, we use $int(x)$ to denote the corresponding integer whose binary representation is $x$. For a nonnegative integer $n$, we use $bin(n)$ to denote the binary representation of $n$. For example, $int(01010111)=87$ and $bin(87)=01010111$.

For an $n$-bit string $x\in\{0,1\}^{n}$ and $i\in[n]$, we use $x[i]$ to denote the $i$th bit of $x$.
If $x$ is a binary representation of an integer, then $x[1]$ is the most significant bit and $x[n]$ is the least significant bit. 
Let $S\in[n]$, we use $x_S$ to denote a new string that concatenates $x$'s bits whose indices are in $S$.
For example, if $x=01010111$ and $S=\{1,3,7,8\}$, then $x_S=0011$.
The Hamming weight of $x$ is denoted by $wt(x)$, that is, $wt(x)$ is the number of 1's in $x$. 

The identity operator is denoted by $I$.
The Kronecker delta $\delta_{i,j}$ is defined by $\delta_{i,j}=1$ if $i=j$ and $\delta_{jk}=0$ if $i\neq j$. 
The indicator function is denoted by $\indicator_S(i)$, which is defined by $\indicator_S(i)=1$ if $i\in S$ and $\indicator_S(i)=0$ if $i\notin S$.

The big $O$ star notation $O^*(\cdot)$ hides the polynomial factors in the standard big $O$ notation.
For example, $8\cdot 2^\frac{n}{2}\cdot n^7 = O^*(2^{\frac{n}{2}})$.
The negligible functions $\negl(\cdot)$ are functions that are smaller than any inverse polynomial: if $\mu(n)\in\negl(n)$ then for all $c\in\mathbb{N}$,  there exists $n_0\in\mathbb{N}$ such that for all $n\ge n_0$, it holds that $\mu(n) < \frac{1}{n^c}$.

\subsection{Quantum computation}
\paragraph{Quantum register.}
We use the sans-serif font to denote quantum registers, \eg $\rgst{A}$, $\rgst{in}$, $\rgst{out}$, $\rgst{anc}$, $\rgst{clock}$.
Throughout this paper, a register consists of qubits.
A qubit may belong to different registers at the same time.
For example, a qubit in the output register $\rgst{out}$ of a quantum circuit is also in the ancilla register $\rgst{anc}$.
By an abuse of notation, sometimes we treat the registers as sets.
We say $\rgst{A}\subseteq\rgst{B}$ if any qubit in $\rgst{A}$ is also a qubit in $\rgst{B}$.
We use $\rgst{A}\cup\rgst{B}$ to denote the register that consists of the qubits that are in $\rgst{A}$ or in $\rgst{B}$, $\rgst{A}\cup\rgst{B}$ to denote the portion that belongs to $\rgst{A}$ and $\rgst{B}$ at the same time, and $\rgst{A}\setminus\rgst{B}$ to denote the portion that belongs to $\rgst{A}$ but not $\rgst{B}$.
When we specify a system (that can be a quantum circuit or a physical qubit system), we use $\overline{\rgst{A}}$ to denote the collection of qubits in the system that are not in $\rgst{A}$.
We use $|\rgst{A}|$ to denote the number of qubits in $\rgst{A}$.
When we specify a register $\rgst{A}$, we use $\rgst{A[i]}$ to denote the $i$th qubit in $\rgst{A}$ for $i\in[|\rgst{A}|]$.

We use $\ket{\psi}\reg{A}$ to denote the register $\rgst{A}$ is in the state $\ket{\psi}$. 
For a Hermitian or a unitary operator $X$, we use $X\reg{A}$ to denote $X$ acting on the register $\rgst{A}$.
When the registers are specified, the order of the tensor product of the operators is not sensitive. For example, $X\reg{A}\otimes Y\reg{B} = Y\reg{B}\otimes X\reg{A}$.
We use $X^{\otimes k}$ to denote the $k$th tensor power of $X$ for $k\in\mathbb{N}$. Namely, $X^{\otimes k}:=\underbrace{X\otimes X\otimes\cdots\otimes X}_{k}$.

\paragraph{Quantum circuit.} 
A quantum circuit is a unitary acting on the union of two disjoint register $\rgst{in}$ and $\rgst{anc}$, which are called input register and ancilla register respectively.
Also, there is an output register $\rgst{out}\subseteq\rgst{in}\cup\rgst{anc}$.
We require the ancilla register is $\ket{0^{n_a}}\reg{anc}$ initially, where $n_a=|\rgst{anc}|$.
That is, if the circuit $U$ takes a quantum state $\ket{\psi}$ as an input, then the final state of the circuit is $U\ket{\psi}\reg{in}\ket{0^{n_a}}\reg{anc}$.
We use the notation $x\gets U(\ket{\psi})$ to denote the event that we obtain the measurement outcome $x\in\{0,1\}^{|\rgst{out}|}$ when we measure on the output register $\rgst{out}$ at the end of the circuit $U$ that takes $\ket{\psi}$ as the input state.
We have $\Pr[x\gets U(\ket{\psi})]=\bra{\psi}\reg{in}\bra{0^{n_a}}\reg{anc}U^\dagger (\proj{x}\reg{out}\otimes I\reg{\overline{out}})U\ket{\psi}\reg{in}\ket{0^{n_a}}\reg{anc}$.

\paragraph{Universal gate set and elementary gates.}
A quantum circuit is implemented by a sequence of quantum gates.
There is a finite set of one-qubit and two-qubit quantum gates called the universal gate set such that for any unitary $U$, we can use the gates in the universal gate sets to implement a quantum circuit arbitrarily close to $U$. 
We call a member in the universal gate set an elementary gate. 
We choose $\{\textsc{hadamard}, \pi/8\textsc{-gate}, \notgate, \cnot\}$ as the universal gate set \cite{NC10}

The $\notgate$ gate acts on a one-qubit register.
The truth table of the $\notgate$ gate is defined by $\notgate\ket{x}=\ket{(1+x)\mod 2}$ where $x\in\{0,1\}$.

We also introduce a multi-control-NOT gate.
\begin{definition}[Multi-control-NOT gate and Toffoli gate]
    A multi-control gate, denoted by $C^k\notgate$, is a quantum gate that acts on $k+1$ qubits.
    Let $\rgst{C}$ be a $k$-qubit register and $\rgst{T}$ be a one-qubit register.
    The truth table of the $C^k\cnot$ gate is defined by 
    \[
    C^k\cnot\ket{x_1,x_2,\dots,x_k}\reg{C}\ket{x_{k+1}}\reg{T}=\ket{x_1,x_2,\dots,x_k}\reg{C}\ket{(x_1x_2\cdots x_k+x_{k+1})\mod 2}\reg{T},
    \]
    where $x_1,x_2,\dots,x_{k+1}\in\{0,1\}$.
    We call $\rgst{C}$ the control qubits and $\rgst{T}$ the target qubit.
    
    For $k=2$, we call $C^2\notgate$ a Toffoli gate.
\end{definition}

According to the definition of $C^k\notgate$, we have that  $\cnot$ is a special case of $C^k\notgate$ gate for $k=1$.

\begin{remark}[Decompose $C^k\notgate$ into Toffolis \cite{NC10}]
    A $C^k\notgate$ gate can be constructed by $2k-3$ Toffoli gates associated with $k-2$ ancilla qubits. The ancilla qubits are in all zero state initially and stay unchanged at the end of circuit.
\end{remark}

\begin{remark}[Decompose Toffoli into elementary gates \cite{NC10}]
    A Toffoli gate is identical to a circuit that acts on three qubits and is composed of 17 elementary gates.
\end{remark}

\subsection{Local Hamiltonian problem} 
A Hamiltonian is an Hermitian operator acting on a quantum register.
Let $H$ be a Hamiltonian, we call the eigenvalues of $H$ the energies of the Hamiltonian.
We use $\lambda(H)$ to denote the smallest eigenvalue of $H$. 
We call $\lambda(H)$ the ground-state energy of $H$, and we call the corresponding eigenstate(s) the ground state(s) of $H$. 
We use $\|H\|$ to denote the operator norm of $H$,
that is, the largest absolute value of the eigenvalues of $H$.

The $k$-local Hamiltonian is defined as follows.
\begin{definition}[Local Hamiltonian]\label{def:klocalH}
We say a Hamiltonian $H$ acting on $n$ qubits is $k$-local if $H$ can be written as $H = \sum_{i=1}^m H_i$ for $m=\poly(n)$ and for all $i$, the following holds.
\begin{itemize}
    \item $H_i$ is a Hamiltonian.
    \item $H_i$ non-trivially acts on at most $k$ qubits.
    \item $\|H_i\|\le \poly(n)$.
\end{itemize} 
\end{definition}

We call $k$ the locality of $H$, and define the local Hamiltonian problem formally.

\begin{definition}[$k$-local Hamiltonian ($\kLH$) problem]
    \label{def:LHP}
    The local Hamiltonian problem is a decision problem that asks whether the ground-state energy of a $k$-local Hamiltonian is greater or less than given energy thresholds.
    \begin{itemize}
        \item \textbf{Inputs:} a $k$-local Hamiltonian $H$ such that acting on $n$ qubits where $k=O(1)$, and two energy thresholds $\EnergyYes,\EnergyNo$ satisfying $\EnergyNo-\EnergyYes\ge \Omega(1)$.
        \item  \textbf{Outputs:}
        \begin{itemize}
            \item YES, if there exists a quantum state $\ket{\psi}\in\mathbb{C}^{2^n}$ such that $\expec{\psi}{H}\le \EnergyYes$.
            \item NO, if $\expec{\psi}{H}\ge \EnergyNo$ for all $\ket{\psi}\in\mathbb{C}^{2^n}$.
        \end{itemize}
    \end{itemize}
    In other words, the $\kLH$ problem asks to decide whether $\lambda(H)\le \EnergyYes$ or $\lambda(H)\ge \EnergyNo$.
    We say an algorithm $\mathcal{A}_{\LH}$ solves $\kLH(H, \EnergyYes, \EnergyNo)$ if $\mathcal{A}_{\LH}$ decides $(H,\EnergyYes,\EnergyNo)$ correctly.
\end{definition}

\subsection{Quantum partition function}

Let us introduce the quantum partition function formally.

\begin{definition}[Approximating quantum partition function of $k$-local Hamiltonian ($\kQPF$) problem]
    \label{def:QPF}
    The quantum partition function problem is to approximate the partition function of a $k$-local Hamiltonian up to a multiplicative error under a certain temperature.
    \begin{itemize}
        \item \textbf{Inputs:} a $k$-local Hamiltonian $H$ acting on $n$ qubits where $k=O(1)$, an inverse temperature $\beta\le \poly(n)$, and an error parameter $\delta\in(0,1)$.
        \item \textbf{Outputs:} $\zout\in\mathbb{R}$ such that
        \begin{equation}
            (1-\delta)Z \leq \zout \leq (1+\delta)Z,
        \end{equation}
        where $Z:= \tr[e^{-\beta H}]$.
        \end{itemize}
        We say an algorithm $\mathcal{A}_{\QPF}$ solves $\kQPF(H, \beta, \delta)$ if $\mathcal{A}_{\QPF}$ outputs such $\zout$ on the inputs $H, \beta, \delta$.
\end{definition}
From now on, when we use the term $\QPF$ or $\kQPF$, it means to approximate the quantum partition function up to a relative error, instead of to compute the exact value.

\subsection{Boolean formula, SAT, SETH and QSETH}
A Boolean variable can be assigned to the value 0 or 1.
We abbreviate the term Boolean variable as variable.
A Boolean formula consists of variables associated with parentheses and logic connectives $\neg$ (NOT), $\vee$ (OR), and $\wedge$ (AND).

Let $\Phi$ be a Boolean formula with $n$ variables $x_1,x_2,\dots,x_n$.
We use an an $n$-bit string $x\in\{0,1\}^n$ to denote an assignment to $x_1,x_2,\dots,x_n$. 
The variable $x_i$ for $i\in[n]$ is assigned to the $i$th bit of $x$.
We use $\Phi(x)$ to denote the value of $\Phi$ when  $x_1,x_2,\dots,x_n$ are assigned to $x$. 
We say $x$ satisfies $\Phi$ if $\Phi(x)=1$.

For all variables $x$, it holds that $\neg(\neg x)= x$.
Let $x_1, x_2,\dots,x_n$ be variables and for each $i\in[n]$, let $\ell_i$ be either $x_i$ or $\neg x_i$. It holds that $\neg(\ell_1\wedge \ell_2\wedge\cdots\wedge \ell_m) = \neg \ell_1 \vee \neg \ell_2 \vee\cdots\vee \neg \ell_m$.

We can use a quantum circuit to compute Boolean formulas.
For $x\in\{0,1\}$, it holds that
\begin{equation}\label{eq:q_comp_not}
    \notgate\ket{x} = \ket{\neg x}.
\end{equation}
For $x_1,x_2,\dots,x_k\in\{0,1\}$, it holds that
\begin{equation}\label{eq:q_comp_and}
    C^k\notgate\reg{C\cup T}\ket{x_1,x_2,\dots,x_k}\reg{C}\ket{0}\reg{T} = \ket{x_1,x_2,\dots,x_k}\reg{C}\ket{x_1\wedge x_2 \wedge \cdots\wedge x_k}\reg{T}.
\end{equation}

We define the conjunctive normal form formula as follows.
\begin{definition}[Conjunctive normal form (CNF) formula]\label{def:kCNF}
Let $\Phi$ be a Boolean formula that consists of $n$ variables $x_1,x_2,\dots, x_k$.
We say $\Phi$ is a $k$CNF formula if $\Phi$ is in the form of $m=\poly(n)$ smaller formulas connected by $\wedge$.
Each smaller formula contains at most $k$ variables, and the variables are connected by $\vee$.

To be more precise, $\Phi=\varphi_1\wedge \varphi_2 \wedge\cdots \wedge \varphi_m$ where $m=\poly(n)$, and for all $i\in [m]$, the following constrains hold.
\begin{itemize}
    \item $\varphi_{i} = (\ell_{i,1}\vee \ell_{i_2}\vee\cdots\vee\ell_{i,k_{i}})$, where $\ell_{i,p}$ can be $x_j$ or $\neg x_j$ for all $p\in[k_i]$ and $j\in[n]$.
    \item $k_i\le k$.
    \item For each $j\in[n]$, $x_j$ and $\neg x_j$ do not appear in $\varphi_i$ at the same time;
    $x_j$, $\neg x_j$ appears in $\varphi_i$ at most once.
\end{itemize}
We say $\varphi_i$ is a clause of $\Phi$ for all $i\in[m]$.
\end{definition}

\begin{definition}[Satisfiability for $\kCNF$ ($\kSAT$) problem]
    \label{def:kSAT}
    The satisfiability problem is a decision problem that asks whether there is an assignment satisfying the given $k$CNF formula.
    \begin{itemize}
        \item \textbf{Inputs:} a $\kCNF$ formula $\Phi$ that contains $n$ variables. The number of clauses of $\Phi$ is $m=\poly(n)$.
        \item \textbf{Outputs:}
        \begin{itemize}
            \item YES, if there exists an assignment $x\in\{0,1\}^n$ such that $\Phi(x)=1$.
            \item No, if $\Phi(x)=0$ for all $x\in\{0,1\}^n$. 
        \end{itemize}
    \end{itemize}
    We say an algorithm $\mathcal{A}_{\SAT}$ solves $\kSAT(\Phi)$ if $\mathcal{A}_{\SAT}$ decides $\Phi$ correctly.
\end{definition}

Though the exact lower bound for $\kSAT$ problem is still unknown, it is widely believed that brute-force search is optimal for classical algorithm and Grover search is optimal for quantum algorithm. The Strong Exponential Time Hypothesis (SETH) \cite{Impagliazzo+01} states that $\kSAT$ needs roughly $2^n$ time for large $k$:
\begin{conjecture}[SETH]\label{conjecture:SETH}
	For all $\varepsilon$, there is some $k\ge 3$ such that $\kSAT$ with $n$ variables and $O(n)$ clauses cannot be solved classically in time $O(2^{(1-\varepsilon)n})$.
\end{conjecture}

A quantum analog was proposed in \cite{Aaronson+CCC20,Buhrman+21} as Quantum Strong Exponential-Time Hypothesis (QSETH).

\begin{conjecture}[Quantum strong exponential time hypothesis (QSETH) \cite{Aaronson+CCC20,Buhrman+21}]\label{conjecture:QSETH}
	For all $\varepsilon$, there is some $k\ge 3$ such that $\kSAT$ with $n$ variables and $O(n)$ clauses cannot be solved quantumly in time $O(2^{(1-\varepsilon)n/2})$.
\end{conjecture}

In the two conjecture above, we can assume the number of clauses can be assumed to be $O(n)$ via the sparsification lemma \cite{Impagliazzo+08}.

\subsection{Fine-grained reduction}\label{sec:fine-grained_reduction}
Here we introduce the fine-grained reduction. 
\begin{definition}[Fine-grained reduction \cite{VassilevskaWilliams15}]\label{def:fine-grain_reduction}
    Let $P$ and $Q$ be two problems and $\mathcal{A}_Q$ be an oracle that solves $Q$ with probability greater than $\frac{2}{3}$. 
    Let $p(\cdot)$ and $q(\cdot)$ be two non-decreasing functions.  
    We say $P$ is $(p,q)$ reducible to $Q$ if for all $\varepsilon$, there exist $\xi$, an algorithm $\mathcal{A}_P$, a constant $d$, and an integer $r(n)$ such that the algorithm $\mathcal{A}_P$ that can black-boxly assess to $\mathcal{A}_Q$ takes an instance of $P$ with size $n$ and satisfies the following.
    \begin{enumerate}
        \item \label{requre1} $\mathcal{A}_P$  solves $P$ with probability greater than $\frac{2}{3}$.
        \item \label{requre2} $\mathcal{A}_P$ runs in $d\cdot p(n)^{1-\xi}$ time.
        \item \label{requre3} $\mathcal{A}_P$ produces at most $r(n)$ instances of $Q$ adaptively.
        \item \label{requre4} $\sum_{i=1}^{r(n)}(q(n_i))^{1-\varepsilon} \le d\cdot (p(n))^{1-\xi}$, where $n_i$ is the size of the $i$th instance of problem $Q$ that is produced by $\mathcal{A}_P$.
    \end{enumerate}
\end{definition}

One can also use a more general reduction, \emph{quantum fine-grained reduction}~\cite{Aaronson+CCC20}.
In a quantum fine-grained reduction, the reduction algorithm $\mathcal{A}_P$ is allowed to query the oracle $\mathcal{A}_Q$ in superposition.
In this work, however, \cref{def:fine-grain_reduction} is sufficient to establish the lower bounds for both the Local Hamiltonian problem and the approximation of the Quantum Partition Function.

When $P$ is $(p,q)$ reducible to $Q$ and every instance produced by $\mathcal{A}_P$ has size $n+o(n)$, we have that if $P$ cannot be solved within time $O(p(n)^{1-\varepsilon})$ for any $\varepsilon$, then $Q$ cannot be solved within time $O(q(n)^{1-\xi})$ for any $\xi$.
\section{Size-preserving circuit-to-Hamiltonian construction}\label{sec:size-preserving circuit-to-Hamiltonian construction}

In this section, we prove the first size-preserving circuit-to-Hamiltonian construction as follows.

\begin{theorem}\label{theorem:main}
    Let $x \in \{0,1\}^n$ be an input and $U_x$ be a $\QMA$ verification quantum circuit $U_x$ acting on $N=\poly(n)$ qubits, and suppose $U_x$ consists of $T = \poly(n)$ gates. Then, for any integer $d\geq1$, there exists a $(d+1)$-local Hamiltonian $H_x$ acting on $N+O(T^{1/d})$ qubits and $\norm{H_x} \leq \poly(n)$ such that for all $\cerr\in(0,1)$, the following two hold.
    \begin{itemize}
        \item If there exists a quantum proof $\ket{\psi}$ such that $\Pr[U_x \,\mathrm{accepts}\, \ket{\psi}] \ge 1-\cerr$, then $\lambda(H_x) \le \cerr$.
        \item If $\Pr[U_x \,\mathrm{accepts}\, \ket{\psi}]\le \cerr$ for any quantum proof $\ket{\psi}$, then $\lambda(H_x) \ge \frac{1}{2}-\cerr$.
    \end{itemize}
\end{theorem}\

To prove the theorem above, we first recall the strategy for the circuit-to-Hamiltonian construction in \cite{Kempe+06}. We then define our clock states and finally prove the statement.

\subsection{The construction by Kitaev, Kempe and Regev \cite{Kempe+06}}\label{subsec:KKR}

In this section, we explain the $2$-local construction from \cite{Kempe+06} with some adaptation. 

We consider a promise problem $(L_{\mathrm{yes}},L_{\mathrm{no}})$ in QMA. 
Let $U=U_T\cdots U_2U_1$ be the quantum circuit of the verifier (which depends on the input), where $T=\poly(n)$. We assume that $U$ acts on $N=\poly(n)$ qubits: the first $m$ qubits contain the proof (the last $N-m$ qubits are initialized to $\ket{0}$). Here, each $U_i$ is a 1-qubit gate or a 2-qubit gate. As in \cite{Kempe+06}, we assume that the only 2-qubit gate used by the circuit is the controlled phase gate, each controlled phase gate is preceded and followed by two Z gates\footnote{Any quantum circuit composed of the standard universal gate set can be compiled to the circuit described above with only linear blowup in circuit size because $\mathrm{CNOT} = (I\otimes H)\mathrm{CZ}(I\otimes H)$.}
Let $T_A\subseteq\{1,\ldots,T\}$ be the set of indices $i$ such that $U_i$ is a $1$-qubit gate.  Let $T_B\subseteq\{1,\ldots,T\}$ be the set of indices $i$ such that $U_i$ is a $2$-qubit gate.

\paragraph{Meta-definition of the Hamiltonian.}\hspace{-4mm}\footnote{This is a ``meta-definition'' since it does not define the concrete implementation of the clock states.}
The Hamiltonian is defined over the Hilbert space $\Hh_{\textrm{circuit}}\otimes \Hh_{\mathrm{clock}}$, where $\Hh_{\textrm{circuit}}$ is a Hilbert space of dimension $2^N$ representing the computation and $\Hh_{\mathrm{clock}}$ is a Hilbert space of dimension at least $T+1$ representing the clock. 
There are $T+1$ orthonormal unit vectors $\ket{\gamma_0},\ldots,\ket{\gamma_{T}}$ in $\Hh_{\mathrm{clock}}$ that are called the legal clock states. Let $\Hh_{\mathrm{legal}}$ denote the subspace spanned by the states $\ket{\gamma_0},\ldots,\ket{\gamma_{T}}$ and $\Pi_{\mathrm{legal}}$ be the projection into this subspace. Let $\Hh_{\mathrm{illegal}}$ be its orthogonal subspace. We define the operators 
 \begin{align*}
 	\stay{t}&=\ket{\gamma_{t}}\bra{\gamma_{t}} 
	\hspace{5mm}\textrm{for any}\,\,t\in\{0,\ldots,T\}\,,\\
 	\add{t}&=\ket{\gamma_{t+1}}\bra{\gamma_{t}}
	\hspace{2mm}\textrm{for any}\,\,t\in\{0,\ldots,T-1\}\,,\\
 	\addd{t}&=\ket{\gamma_{t+2}}\bra{\gamma_{t}}
	\hspace{2mm}\textrm{for any}\,\,t\in\{0,\ldots,T-2\}\,,
 \end{align*}
 acting on $\Hh_{\mathrm{clock}}$
 that respectively stop the clock, increment the clock by one step and increment the clock by two steps.

Consider the following Hamiltonian:
\[
	H = \alpha_{\mathrm{in}} H_{\mathrm{in}} + \alpha_{\mathrm{out}}H_{\mathrm{out}} + \alpha_{A}\sum_{t\in T_A} H_{\mathrm{prop},t} + \alpha_{B}\sum_{t\in T_B} (H_{\mathrm{qubit},t} +H_{\mathrm{time},t})+ \alpha_{\mathrm{clock}}I\otimes H_{\mathrm{clock}}\,.
\]
The choice of the coefficients $\alpha_{\mathrm{in}}$, $\alpha_{\mathrm{out}}$, $\alpha_{A}$, $\alpha_{B}$, $\alpha_{\mathrm{clock}}$ will be discussed later. The terms $H_{\mathrm{in}}$ and $H_{\mathrm{out}}$ are defined as follows:
\[
H_{\mathrm{in}} =
\sum_{i=m+1}^{N} \ket{1}_i \bra{1}_i \;\otimes\; \stay{0}\,,
\hspace{5mm}
H_{\mathrm{out}} = \ket{0}_1 \bra{0}_1 \;\otimes\; \stay{T}\,.
\]
For any $t\in T_A$,
\[
H_{\mathrm{prop},t} =
\frac{1}{2}\Bigl(
I \otimes \stay{t}
+ I \otimes \stay{t-1}
- U_t \otimes \add{t-1}
- U_t^\dagger \otimes \add{t-1}^\dagger
\Bigr)\,.
\]
For any $t\in T_B$,
\begin{align*}
H_{\mathrm{time},t}
&=
\frac{1}{8}\, I \otimes \Bigl(
\stay{t}
+ 6 \stay{t+1} 
+ \stay{t+2}
+ 2 \addd{t}
+ 2 \addd{t}^\dagger
+ \add{t}+\add{t}^\dagger+ \add{t+1}+\add{t+1}^\dagger
+\stay{t-3}
\\
&\hspace{22mm}
+ 6 \stay{t-2}
+ \stay{t-1}
+ 2 \addd{t-1}+2 \addd{t-1}^\dagger
+ \add{t-3}+\add{t-3}^\dagger+\add{t-1}+\add{t-1}^\dagger
\Bigr)\,,\\
H_{\mathrm{qubit},t}&=\frac{1}{2}\left(-2\ket{0}\bra{0}_{f_t}-2\ket{0}\bra{0}_{s_t}+\ket{1}\bra{1}_{f_t}+\ket{1}\bra{1}_{s_t}\right)\otimes(\add{t-1}+\add{t-1}^\dagger)\,,
\end{align*} 
where $f_t$ and $s_t$ denote the first and second qubits of the controlled phase gate $U_t$.
The term $H_{\mathrm{clock}}$ will be chosen later so that the following two conditions are satisfied
\begin{align}
	\norm{H_{\mathrm{clock}}}&\le\poly(n)\,,\tag{C1}\\
	\bra{\gamma_t}H_{\mathrm{clock}}\ket{\gamma_t}&=0
	 \hspace{6mm}\textrm{ for any }\,\,t\in \{0,\ldots,T\}\,,\tag{C2}\\
	 \bra{\varphi}H_{\mathrm{clock}}\ket{\varphi}&\ge 1 
	\hspace{4mm}\textrm{ for any }\,\,\ket{\varphi}\in\Hh_{\textrm{illegal}}\,.\tag{C3}
\end{align}

\paragraph{Concrete implementation.}
A straightforward implementation uses a binary clock, as in the original construction by \cite{Kitaev+02}. While requiring only $O(\log T)$ clock qubits (and no $H_{\mathrm{clock}}$ term), the locality of the Hamiltonian would become $O(\log n)$. We discuss below how to reduce the locality. 

The basic idea is to replace each of the operator $\stay{t},\add{t},\addd{t}$ by another operator $\tstay{t},\tadd{t},\taddd{t}$ that has better locality. In order for the implementation to be correct, the following three conditions need to be satisfied. 
\begin{align}
	\textrm{For any }\,\,t\in \{0,\ldots,T\},&\,\,\,
	 \begin{cases}
		\tstay{t}\ket{\gamma_{t}}=\ket{\gamma_{t}} &  \\
		\Pi_{\textrm{legal}}\tstay{t}\ket{\varphi}=0  & \text{ for any $\ket{\varphi}\in\Hh_{\textrm{clock}}$ orthogonal to $\ket{\gamma_{t}}$}\,.
	\end{cases}\tag{C4}\\
	%%%%
	\textrm{For any }\,\,t\in \{0,\ldots,T-1\},&\,\,\,
	\begin{cases}
		\tadd{t}\ket{\gamma_t}=\ket{\gamma_{t+1}} &  \\
		\Pi_{\textrm{legal}}\tadd{t}\ket{\varphi}=0  & \text{ for any $\ket{\varphi}\in\Hh_{\textrm{clock}}$ orthogonal to $\ket{\gamma_{t}}$}\,.
	\end{cases}\tag{C5}\\
	%%%%%
	\textrm{For any }\,\,t\in \{0,\ldots,T-2\},&\,\,\,
	\begin{cases}
		\taddd{t}\ket{\gamma_{t}}=\ket{\gamma_{t+2}} &  \\
		\Pi_{\textrm{legal}}\taddd{t}\ket{\varphi}=0  & \text{ for any $\ket{\varphi}\in\Hh_{\textrm{clock}}$ orthogonal to $\ket{\gamma_{t}}$}\,.
	\end{cases}\tag{C6}
\end{align}
The first part of each condition says that each $\tstay{t},\tadd{t},\taddd{t}$ behaves as $\stay{t},\add{t},\addd{t}$, respectively, on the clock state $\ket{\gamma_t}$. The second part says that while the behavior may differ from the behavior of $\stay{t},\add{t},\addd{t}$ on other states, $\tstay{t},\tadd{t},\taddd{t}$ never maps an illegal state or even a legal state other than $\ket{\gamma_t}$ to a legal state.

The analysis based on the projection lemma (Lemma 1 in \cite{Kempe+06}) shows the following.
\begin{proposition}[Adapted from Section 5 in \cite{Kempe+06}]\label{prop:KKR}
	If Conditions $(\textrm{C}_1)\sim(\textrm{C}_6)$ are satisfied, coefficients $\alpha_{\mathrm{in}}$, $\alpha_{\mathrm{out}}$, $\alpha_{A}$, $\alpha_{B}$, $\alpha_{\mathrm{clock}}$ can be chosen in the range $[0,\poly(n)]$ so that the following holds:
	\begin{itemize}
        \item If there exists a quantum proof $\ket{\psi}$ such that $\Pr[U_x \,\mathrm{accepts}\, \ket{\psi}] \ge 1-\varepsilon$, then $\lambda(H_x) \le \varepsilon$.
        \item If $\Pr[U_x \,\mathrm{accepts}\, \ket{\psi}]\le \varepsilon$ for any quantum proof $\ket{\psi}$, then $\lambda(H_x) \ge \frac{1}{2}-\varepsilon$.
    \end{itemize}
\end{proposition}

\paragraph{Implementation by Kempe-Kitaev-Regev.}
The implementation in \cite{Kempe+06} uses the unary encoding 0000, 1000, 1100, 1110, $1111\cdots$ for the clock states, i.e., uses the $T$-qubit clock states
\[
\ket{\gamma_t}=\ket{1^t0^{T-t}}
\]
for each $t\in\{0,\ldots,T\}$, and the operators $\tstay{t},\tadd{t},\taddd{t}$ described in the following table:\vspace{3mm}

\begin{table}[h!t]
\centering
\renewcommand{\arraystretch}{1.4}
\begin{tabular}{|l|l|l|l|}
	\hline
	Operator&$t$&Implementation& Locality \\\hline
	\multirow{3}{*}{$\stay{t}$}&
	$t=0$ & $\tstay{0}=\ket{0}_1\bra{0}_1$ & 1\\\cline{2-4}
	&$t\in\{1,\ldots,T-1\}$  & $\tstay{t}=\ket{10}_{t,t+1}\bra{10}_{t,t+1}$ & 2\\\cline{2-4}
	&$t=T$ & $\tstay{T}=\ket{1}_T\bra{1}_T$ & 1\\\hline
	$\add{t}$ &$t\in\{0,\ldots,T-1\}$  & $\tadd{t}=\ket{1}_{t+1}\bra{0}_{t+1}$ & 1\\\hline
	$\addd{t}$ &$t\in\{0,\ldots,T-2\}$  & $\taddd{t}=\ket{11}_{t+1,t+2}\bra{00}_{t+1,t+2}$ & 2\\\hline
\end{tabular}
\caption{Concrete implementation of operators in \cite{Kempe+06}}\label{tab:operators in KKR06}
\end{table}

The term $H_{\mathrm{clock}}$ used in~\cite{Kempe+06} is defined as follows:
\[
H_{\mathrm{clock}}
=
\sum_{1 \le i < j \le T}
\ket{01}_{ij}\bra{01}_{ij}\,.
\]

This implementation satisfies Conditions $(\textrm{C}_1)\sim(\textrm{C}_6)$. With this implementation, 
each term $H_{\mathrm{in}}$, $H_{\mathrm{out}}$, $H_{\mathrm{prop},t}$, $H_{\mathrm{qubit},t}$, $H_{\mathrm{time}, t}$, $H_{\mathrm{clock}}$ has locality 2. The whole Hamiltonian is thus 2-local.

\subsection{$(n,d)$-clock state from \cite{Chia+CCC23} and associated operators to implement LH}

In this subsection, we recall an $(n,d)$-clock state proposed in \cite{Chia+CCC23}, and introduce and analyze operators to use it for circuit-to-Hamiltonian constructions.
The $(n,d)$-clock state is based on a Johnson graph.
The vertices of the Jonson graph are all the subsets of $[n]$ with size $k<n$.
Two vertices are adjacent if and only if the subsets share exactly $d-1$ elements.

\begin{definition}[Johnson graph]
    \label{def:Johnson}
    Let $n,d\in\mathbb{N}$ and $d<n$. We say a Johnson graph $J(n,d)=(V,E)$ is a graph that satisfies the following requirement.
    \begin{itemize}
        \item $V=\{S\subseteq[n]:|S|=d\}$.
        \item $E=\{(S,S')\in V^2:|S\cap S'|=d-1\}$.
    \end{itemize}
\end{definition}
The number of vertices of $J(n,d)$ is $\binom{n}{d}$.
In \cite{Alspach12}, it has been proved that for all $n,d$, there is a Hamiltonian path\footnote{A Hamiltonian path is a path that visits every vertex in a graph exactly once. Please do not confuse with the physical quantity (local Hamiltonian) $H$.} in the Johnson graph $J(n,d)$.
Also, the proof implicitly construct an algorithm that finds the Hamiltonian path in $O(n^d)$ time. 

Now we explain how to construct the clock.
A clock state is an $n$-qubit state.
Each clock state $\ket{\clock{t}}$ corresponds to a vertex $S_t$ on a Hamiltonian path in $J(n,d)$.
All the clock states are basis states of the computational basis whose number of 1's equals to $d$.
The $i$th qubit in $\ket{\clock{t}}$ is $\ket{1}$ if and only if $i$ is chosen into $S_t$.
In other words, when we treat the clock register as the set $[n]$ and let $\rgst{S_t}$ be the portion therein corresponding to $S_t$, the state in $\rgst{S_t}$ is $\ket{1^d}$.
The next clock $\ket{\clock{t+1}}$ corresponds to the vertex $S_{t+1}$ that is adjacent to $S_{t}$ in the Hamiltonian path.
To update the $\ket{\clock{t}}$ to $\ket{\clock{t+1}}$, we make $d-1$ of 1's unchanged and flip two bits. Hence the update of the clock state can be done with constant local operation.
We define the clock state formally as follows.
\begin{definition}[$(n,d)$-clock state \cite{Chia+CCC23}]\label{def:Yao-Ting clock}
    A $(n,d)$-clock state is a collection of $n$-qubit quantum state $\{\ket{\gamma_t}\}_{t=0}^T$ defined by a Johnson graph $J(n,d)=(V,E)$ and $T=\binom{n}{d}-1$. 
    Let $S_0, S_1,\dots, S_{T}\in V$ and the sequence $S_0, S_1,\dots, S_{T}\in V$ forms a Hamiltonian path in $J(n,d)$.
    For any $t\in\{0\}\cup [T]$ the state $\ket{\gamma_t}$ is defined by  
    $\ket{\gamma_t}:=\bigotimes_{i\in[n]}\ket{\indicator_{S_t}(i)}$.
\end{definition}

To use the clock above for the circuit-to-Hamiltonian construction, we introduce several local operators to pause and move time on the clock.
For all $t\in 0\cup [T]$, let $\rgst{S_t}$ be the corresponding register of $S_t$.
For all $t\in [T]$, we define the operator $F_t$ as follows.
\begin{equation}
    \label{eq:forward}
    F_t:=\ketbra{1}{0}\reg{S_{t-1}\setminus S_{t}}\otimes \ketbra{0}{1}\reg{S_{t}\setminus S_{t-1}}\otimes \proj{1^{d-1}}\reg{S_{t}\cap S_{t-1}}.    
\end{equation}
Hence we have 
\begin{equation}
    \label{eq:backward}
    F_{t}^\dagger:=\ketbra{0}{1}\reg{S_{t-1}\setminus S_{t}}\otimes \ketbra{1}{0}\reg{S_{t}\setminus S_{t}}\otimes \proj{1^{d-1}}\reg{S_{t}\cap S_{t-1}}.    
\end{equation}
It holds that
\begin{equation}
    \label{eq:forward_a_step}
    F_{t}\ket{\gamma_{t'}}=\delta_{t',t-1}\ket{\gamma_{t'+1}},
\end{equation}
and
\begin{equation}
    \label{eq:backward_a_step}
    F_{t}^{\dagger}\ket{\gamma_{t'}}=\delta_{t',t}\ket{\gamma_{t'-1}}.
\end{equation}
That is, $F_{t}$ ``forwards'' the clock $\ket{\clock{t-1}}$ one step, and eliminate all the other state $\ket{\clock{t'}}$ where $t'\neq t-1$.
Likewise, $F_{t}^\dagger$ ``backwards'' the clock $\ket{\clock{t}}$ one step, and eliminate all the other state $\ket{\clock{t'}}$ where $t'\neq t$.
We have that $F_{t}$ is $(d+1)$-local.
Note that $F_{t}$ is neither unitary nor Hermitian.

For $t\in\{0\}\cup [T]$, we define the operator $P_t$ as following.
\begin{equation}
    \label{eq:pause}
    P_{t}:= \proj{1^d}\reg{S_t}.
\end{equation}
We have $P_{t}^\dagger = P_{t}$.
We also have
\begin{equation}
    \label{eq:pause_a_step}
    P_{t}\ket{\gamma_{t'}}=\delta_{t',t}\ket{\gamma_{t'}}.
\end{equation}
That is, $P_{t}$ ``pauses'' the clock $\ket{\clock{t}}$, and eliminate all the other state $\ket{\clock{t'}}$ where $t'\neq t$.
The operator $P_{t}$ is $d$-local.

We have seen that we can ``forward'', ``backword'' and ``pause'' the clock with local operators. To use the $(n,d)$ clock for the circuit-to-Hamiltonian construction, we also want valid clock states $\{\ket{\gamma_t}\}_{t=0}^T$ to be the ground state for the Hamiltonian, and give some energy penalty to any other states orthogonal to valid clock states. In order to achieve this, we define
\begin{equation}\label{eq:H_stab}
    H_{\mathrm{stab}} :=  \binom{n}{d} H_{>d} + H_{<d} - \left(\binom{n}{d} - 1\right) I,
\end{equation}
    where
\begin{align}
    H_{>d} &:=\sum_{\rgst{S}:|\rgst{S}|=d+1} \ketbra{1^{d+1}}{1^{d+1}}\reg{S},\label{eq:H_more}\\
    H_{<d} &:=\sum_{\rgst{S}:|\rgst{S}|=d}\sum_{x\in\{0,1\}^{d}\setminus\{1^d\}} \proj{x}\reg{S},\label{eq:H_fewer}
\end{align}

The term $H_{>d}$ gives penalty to the state that contains ``too many'' 1's, and the term $H_{<d}$ gives penalty to the state that contains ``too few'' 1's.

\begin{lemma}
    The linear space for the ground states of $H_{\mathrm{stab}}$ is spanned by $\{\ket{\gamma_t}\}_{t=0}^T$, and for any state $\ket{\phi}$ orthogonal to all $\{\ket{\gamma_t}\}_{t=0}^T$, we have
    \[
        \expec{\phi}{H_{\mathrm{stab}}}\ge 1.
    \]
\end{lemma}

\begin{proof}
    Since $H_{\mathrm{stab}}$ is diagonalized in the computational basis i.e., each computational state is an eigenstate of $H_{\mathrm{stab}}$. To verify the above statement, we can check $\expec{x}{H_{\mathrm{stab}}}$ for all $x\in\{0,1\}^{n}$.

We divide the Hilbert space $\mathcal{H}_\mathrm{clock}$ into three subspaces $\hs{L}_{=d}$, $\hs{L}_{>d}$, and $\hs{L}_{<d}$ which are defined below.
    \begin{align}
        \hs{L}_{<d} &: =\spn(\{w\in\{0,1\}^{n}:wt(w)<d\}),\\
        \hs{L}_{=d} &: =\spn(\{y\in\{0,1\}^{n}:wt(y)=d\}),\\
        \hs{L}_{>d} &: =\spn(\{z\in\{0,1\}^{n}:wt(z)>d\})
    \end{align}
Note that $\hs{L}_{=d}$ is the subspace spanned by the valid clock states $\{\ket{\gamma_t}\}_{t=0}^T$.

We first calculate $H_{>d}$ acting on the states in $\hs{L}_{=d}$, $\hs{L}_{>d}$, and $\hs{L}_{<d}$ respectively.
    For any $\rgst{S}$ with size $d+1$, $\proj{1^{d+1}}\reg{S}\ket{x} = \ket{x}$ if the subset-string $x_S=1^{d+1}$.
    Hence, we have the following.
    \begin{itemize}
        \item For all $w\in\{0,1\}^{n}$ such that $wt(w)<d$, it holds that $\bra{w}H_{>d}\ket{w} = 0$ 
        because for all $S\subseteq[n]$ with size $d+1$, $z_S$ has at most $d-1$ of 1's.
        \item For all $y\in\{0,1\}^{n}$ such that $wt(y)=d$, it holds that $\bra{y}H_{>d}\ket{y} = 0$ because for all $S\subseteq[n]$ with size $d+1$, $w_S$ has at most $d$ of 1's.
        \item For all $z\in\{0,1\}^{n}$ such that $wt(z)>d$, it holds that $\bra{z}H_{>d}\ket{z} = r$ where $r\ge 1$, because there is at least one $S\subseteq [n]$ with size $k+1$ such that $y_S=1^{d+1}$. 
    \end{itemize}

Then we calculate $H_{<d}$ acting on the states in $\hs{L}_{=d}$, $\hs{L}_{>d}$, and $\hs{L}_{<d}$ respectively.

    Let $w\in \{0,1\}^{n}$ such that $wt(w)<d$. 
    For any $\rgst{S}$ with size $d$, we have $\sum_{x\in\{0,1\}^d\setminus\{1^d\}}\proj{x}\reg{S}\ket{w} = \ket{w}$, because there exist exactly one $x\in\{0,1\}^d\setminus\{1^d\}$ such that $w_{S}=x$.
    Hence, \[\sum_{\rgst{S}:|\rgst{S}|=d}\sum_{x\in\{0,1\}^d\setminus\{1^d\}}\proj{x}\reg{S}\ket{z} = \binom{n}{d} \ket{z}.\]

    Let $y\in\{0,1\}^{n}$ such that $wt(y)=d$. 
    Let $S'\subseteq[n]$ be the subset with size $d$ such that $w_{S'}=1^d$.
    For any $\rgst{S}$ with size $d$ such that $S\neq S'$, we have that $\sum_{x\in\{0,1\}^d\setminus\{1^d\}}\proj{x}\reg{S}\ket{y} = \ket{y}$, and
    for $S = S'$, we have $\sum_{x\in\{0,1\}^d\setminus\{1^d\}}\proj{x}\reg{S}\ket{y} = 0$;
    Hence, \[\sum_{\rgst{S}:|\rgst{S}|=d}\sum_{x\in\{0,1\}^d\setminus\{1^d\}}\proj{x}\reg{S}\ket{w} = \left(\binom{\nclock}{d} -1 \right)\ket{w}.\]

    For $z\in\{0,1\}^{n}$ such that $wt(z)>d$,
    since the penalty has been given by $H_{>d}$ already,
    $H_{<d}\ket{z} = p\ket{z}$ where $p\ge 0$ is sufficient for us.
    
    Therefore, we have the follows.
    \begin{itemize}
        \item For all $w\in\{0,1\}^{n}$ such that $wt(w)<d$, it holds that $H_{<d}\ket{w} = \binom{n}{d} \ket{w}$.
        \item For all $y\in\{0,1\}^{n}$ such that $wt(y)=d$, it holds that $H_{<d}\ket{y} = \left(\binom{n}{d}-1\right) \ket{y}$.
        \item For all $z\in\{0,1\}^{n}$ such that $wt(z)>d$, it holds that $H_{<d}\ket{z} =p\ket{z}$ where $p\ge 0$.
    \end{itemize}
    
    As a result, we have the follows.
    \begin{itemize}
        \item For all $w\in\{0,1\}^{n}$ such that $wt(w)<d$, it holds that $\bra{w}H_{\mathrm{stab}}\ket{w} = 1$.
        \item For all $y\in\{0,1\}^{n}$ such that $wt(y)=d$, it holds that $\bra{y}H_{\mathrm{stab}}\ket{y} = 0$.
        \item For all $z\in\{0,1\}^{n}$ such that $wt(z)>d$, it holds that $\bra{z}H_{\mathrm{stab}}\ket{z} \ge 1$.
    \end{itemize}
    
    Consequently, if $\ket{\phi}\in\hs{L}_{=d}$, then 
    $\expec{\phi}{H_{\mathrm{stab}}}=0$, and if $\ket{\phi}\in\hs{L}_{=d}^\perp = \hs{L}_{<d} \oplus \hs{L}_{>d}$, then $\expec{\phi}{H_{\mathrm{stab}}}\ge 1$.
\end{proof}

%======
\subsection{Our new clock states and proof of \cref{theorem:main}}
%======

We show how to modify the construction from \cref{subsec:KKR} to obtain a length-preserving construction. The construction is the same as the construction in \cref{subsec:KKR}, except that we are using a new clock, new clock operators, and a new term $H_{\mathrm{clock}}$. 

\begin{proof}[Proof of \cref{theorem:main}]
If $d=1$, the statement follows from the construction of \cite{Kempe+06} in \cref{subsec:KKR}. For the rest of the proof, we assume $d\geq2$.

To further reduce the locality, we introduce a new clock to count from $0$ to~$T$. We assume (without loss of generality) that $T$ is of the form $T=\binom{\cmax}{d-1}(\cmax+1)-1$ for some $\cmax$. We use two clocks, each using $\cmax$ qubits. The first clock is the $(a,d-1)$-clock following \cref{def:Yao-Ting clock}, and the second clock is the unary clock on $a$ qubits. Parameters are slightly different from the previous section, and not to confuse readers, we clarify the operator of the $(a,d-1)$-clock states.

\begin{equation}
    \label{eq:forward_final}
    F'_t:=\ketbra{1}{0}\reg{S_{t-1}\setminus S_{t}}\otimes \ketbra{0}{1}\reg{S_{t}\setminus S_{t-1}}\otimes \proj{1^{d-2}}\reg{S_{t}\cap S_{t-1}}.    
\end{equation}

\begin{equation}
    \label{eq:pause_final}
    P'_{t}:= \proj{1^{d-1}}\reg{S_t}.
\end{equation}

\begin{equation}\label{eq:H_stab_final}
    H'_{\mathrm{stab}} :=  \binom{a}{d-1} H'_{>d-1} + H'_{<d-1} - \left(\binom{a}{d-1} - 1\right) I,
\end{equation}
    where
\begin{align}
    H'_{>d-1} &:=\sum_{\rgst{S}:|\rgst{S}|=d} \ketbra{1^{d}}{1^{d}}\reg{S},\label{eq:H_more_final}\\
    H'_{<d-1} &:=\sum_{\rgst{S}:|\rgst{S}|=d-1}\sum_{x\in\{0,1\}^{d-1}\setminus\{1^{d-1}\}} \proj{x}\reg{S},\label{eq:H_fewer_final}
\end{align}

The whole clock thus requires only $2\cmax\le 2T^{1/d}$ qubits. To represent the order of the $T+1$ representations of the clock, we define a map \[
	g\colon\{0,\ldots, T\}\to\{1,\ldots, \binom{\cmax}{d-1}\}\times\{0,\ldots, \cmax\}\,
\] 
as $g(t) = (1+\floor{t/(a+1)},\min(m,2a+1-m))$ where $m=t \bmod (2a+2)$. We implicitly use $g$ to interpret time steps $t\in\{0,\ldots, T\}$ as pairs $(t_1,t_2)\in\{1,\ldots, \binom{\cmax}{d}\}\times\{0,\ldots, \cmax\}$.
We explain the construction of $g$ for $\cmax=3,d=1$ in \cref{table}.

\begin{table}[h!t]
\centering
\renewcommand{\arraystretch}{1.4}
\begin{tabular}{c|c|c|}
	$t$&$g(t)=(t_1,t_2)$& clock state\\\hline
	$t= 0$&$(1,0)$&$\red{\ket{100}}\otimes\blue{\ket{000}}$\\
	$t= 1$&$(1,1)$&$\red{\ket{100}}\otimes\blue{\ket{100}}$\\
	$t= 2$&$(1,2)$&$\red{\ket{100}}\otimes\blue{\ket{110}}$\\
	$t= 3$&$(1,3)$&$\red{\ket{100}}\otimes\blue{\ket{111}}$\\
	$t= 4$&$(2,3)$&$\red{\ket{010}}\otimes\blue{\ket{111}}$\\
	$t= 5$&$(2,2)$&$\red{\ket{010}}\otimes\blue{\ket{110}}$\\
	$t= 6$&$(2,1)$&$\red{\ket{010}}\otimes\blue{\ket{100}}$\\
	$t= 7$&$(2,0)$&$\red{\ket{010}}\otimes\blue{\ket{000}}$\\
	$t=8$&$(3,0)$&$\red{\ket{001}}\otimes\blue{\ket{000}}$\\
	$t=9$&$(3,1)$&$\red{\ket{001}}\otimes\blue{\ket{100}}$\\
	$t=10$&$(3,2)$&$\red{\ket{001}}\otimes\blue{\ket{110}}$\\
	$t=11$&$(3,3)$&$\red{\ket{001}}\otimes\blue{\ket{111}}$
\end{tabular}\caption{Our clock for $T=11$ (i.e., $\cmax=3$). The part in red represents the first clock, while the part in blue represents the second clock. The first clock uses the encoding 100, 010, 001. The second clock uses the encoding 000, 100, 110, 111. }\label{table}
\end{table}

Our implementation of the clock operators is described in \cref{tab:clock_implementation}.

{
\begin{table}[h!]
    \centering
\renewcommand{\arraystretch}{1.3}
\begin{tabular}{|l|l|l|l|l|}
	\hline
	Operator&$t_1$&$t_2$&Implementation & Locality \\\hline
	\multirow{3}{*}{$\stay{t}$} &
	\multirow{3}{*}{any}&
	$t_2=0$  & 
	$\red{P'_{t_1}} \otimes\blue{\ket{0}_1 \bra{0}_1}$
	& $d$\\\cline{3-5}
	& 
	&
	$t_2\in\{1,\ldots,\cmax-1\}$  &$\red{P'_{t_1}} \otimes\blue{\ket{10}_{t_2,t_2+1} \bra{10}_{t_2,t_2+1}}$ & $d+1$\\\cline{3-5}
	& 
	&
	$t_2=\cmax$  &
	$\red{P'_{t_1}} \otimes \blue{\ket{1}_{\cmax} \bra{1}_{\cmax}}$ 
	& $d$\\\hline
	%%%%%%%%%%%%%%%%%%%%%%%%%%%%%%%%%%%%%%
	\multirow{4}{*}{$\add{t}$}
	& 
	\multirow{2}{*}{odd}&
	$t_2\in\{0,\ldots,\cmax-1\}$  & 
	$\red{P'_{t_1}} \otimes  \blue{\ket{1}_{t_2+1}\bra{0}_{t_2+1}}$ & $d$\\\cline{3-5}
	& 
	&
	$t_2=\cmax$  \,\,(if $t_1\neq \binom{a}{d-1}$)& 
	$\red{F'_{t_1+1}}\otimes  \blue{\ket{1}_{\cmax}\bra{1}_{\cmax}}$
    & $d+1$\\\cline{2-5}
	%$\,\,(t<T)$
	& 
	\multirow{2}{*}{even}&
	$t_2\in\{1,\ldots,\cmax\}$  & 
	$\red{P'_{t_1}}\otimes\blue{\ket{0}_{t_2}\bra{1}_{t_2}}$ & $d$\\\cline{3-5}
	& 
	&
	$t_2=0$  & 
	$\red{F'_{t_1+1}} \otimes \blue{\ket{0}_{1}\bra{0}_{1}}$ & $d+1$\\\cline{1-5}
	%%%%%%%%%%%%%%%%%%%%%%%%%%%%%%%%%%%%%
	%%%%%%%%%%%%%%%%%%%%%%%%%%%%%%%%%%%%%%
	\multirow{6}{*}{$\addd{t}$}
	& 
	\multirow{3}{*}{odd}&
	$t_2\in\{0,\ldots,\cmax-2\}$ & 
	$\red{P'_{t_1}}
	\otimes  \blue{\ket{11}_{t_2+1,t_2+2}\bra{00}_{t_2+1,t_2+2}}$ & $d+1$\\\cline{3-5}
	& 
	&
	$t_2=\cmax-1$  \,\,(if $t_1\neq \cmax$)& 
	$\red{F'_{t_1+1}}\otimes  \blue{\ket{0}_{\cmax}\bra{1}_{\cmax}}$ & $d+1$\\\cline{3-5}
	& 
	&
	$t_2=\cmax$  \,\,(if $t_1\neq \cmax$)& 
	$\red{F'_{t_1+1}} \otimes \blue{\ket{1}_{\cmax}\bra{0}_{\cmax}}$ & $d+1$\\\cline{2-5}
	%$\,\,(t<T-1)$
	& 
	\multirow{3}{*}{even}&
	$t_2\in\{2,\ldots,\cmax-2\}$  & 
	$\red{P'_{t_1}}\otimes\blue{\ket{00}_{t_2-1,t_2}\bra{11}_{t_2-1,t_2}}$ & $d+1$\\\cline{3-5}
	& 
	&
	$t_2=1$  & 
	$\red{F'_{t_1+1}} \otimes \blue{\ket{0}_{1}\bra{1}_{1}}$ & $d+1$\\\cline{3-5}
	& 
	&
	$t_2=0$  & 
	$\red{F'_{t_1+1}}  \otimes  \blue{\ket{1}_{1}\bra{0}_{1}}$ & $d+1$\\\cline{1-5}
	%%%%%%%%%%%%%%%%%%%%%%%%%%%%%%%%%%%%%
\end{tabular}
\caption{The implementation of clock operations. The part in red acts on the qubits corresponding to the first clock and the part in blue acts on the qubits corresponding to the second clock.}
    \label{tab:clock_implementation}
\end{table}
}

\clearpage
\paragraph{The term $\bm{H_{\mathrm{clock}}}$.}

The term $H_{\mathrm{clock}}$ is defined as follows:
\[
	H_{\mathrm{clock}}
	=
    H'_{\mathrm{stab}}\otimes I
	+
	\sum_{0 \le i < j \le \cmax}
	I\otimes\ket{01}_{ij}\bra{01}_{ij}\,,
\]
where in term, the left part acts on the register corresponding to the first clock and the right part acts on the qubits corresponding to the second clock.

\paragraph{Analysis of the whole Hamiltonian.}
Without loss of generality, we assume that
\begin{equation}
	\textrm{$U_t=I$ for any $t$ of the form $t=(t_1,0)$ with $t_1$ even or $t=(t_1,a)$ with $t_1$ odd.}\tag{$\star$}
\end{equation}
The assumption can be fulfilled by adding the dummy gates $I$ at such time steps. 
We add those time steps such that $U_t=I$ to $T_A$, i.e., we consider the gates $U_t=I$ as 1-qubit gates.
The locality of each term of the whole Hamiltonian is represented in the following table (note that we crucially use Assumption ($\star$) to guarantee that the locality of $H_{\mathrm{prop},t}$ and $H_{\mathrm{qubit},t}$ is $d+1$): \vspace{3mm}

{
\renewcommand{\arraystretch}{1.3}
\begin{tabular}{|l|l|}
	\hline
	Term & Locality \\\hline
	$H_{\mathrm{in}}$&$1+d=d+1$\\\hline
	$H_{\mathrm{out}}$&$1+d=d+1$\\\hline
	$H_{\mathrm{prop},t}$&$\max\{d+1,1+d\}=d+1$\\\hline
	$H_{\mathrm{qubit},t}$&$1+d=d+1$\\\hline
	$H_{\mathrm{time}, t}$&$\max\{d,d+1\}=d+1$\\\hline
	$H_{\mathrm{clock}}$& $d$\\\hline
\end{tabular}
}\vspace{5mm}

The whole Hamiltonian is thus $(d+1)$-local. We can easily check that Conditions $(\textrm{C}_1)\sim(\textrm{C}_6)$ are satisfied. From \cref{prop:KKR}, we conclude that the circuit-to-Hamiltonian construction is correct. From the condition $(\star)$, $T$ is $\binom{a}{d-1}a$. Therefore, we can take $a=O(T^{1/d})$ and the number of qubits $H_x$ acting on is $N+O(T^{1/d})$.
\end{proof}
\section{Fine-grained hardness for local Hamiltonian problem}\label{sec:SETH-hardness of LH}
In this section, we are going to prove the lower bound for the 3-LH problem.

\begin{theorem}[Lower bound for $3\textrm{-}\LH$]\label{thm:SETH hardness for 3LH}
    Assume that SETH (resp. QSETH) holds.
    Then, for any constant $\xi>0$, for any classical (resp. quantum) algorithm $\mathcal{A}_{LH}$, there exists a $3$-local Hamiltonian $H$ acting on $n_H$ qubits associated with $\EnergyYes,\EnergyNo$ satisfying $\EnergyNo-\EnergyYes \geq \Omega(1)$ such that $\mathcal{A}_{\LH}$ cannot decide $\lambda(H)\geq \EnergyNo$ or $\lambda(H)\leq \EnergyYes$ with probability greater than $2/3$ in $O(2^{n_H(1-\xi)})$ (resp. $O(2^{n_H(1-\xi)/2})$) time.
\end{theorem}

We will prove \cref{thm:SETH hardness for 3LH} in \cref{sec:Hlowerbound}.
To prove it, we will use the size-preserving circuit-to-Hamiltonian construction (\cref{theorem:main}) associated with the following lemma.

\begin{lemma}[A quantum circuit can compute $\kCNF$]
    \label{lem:U_compute_Phi}
    For any positive integer $k$ and $c>0$, for all $n\in\mathbb{N}$, for any $\kCNF$ formula $\Phi$ that contains $n$ variables and $m=O(n^c)$ clauses, there exists a quantum circuit $U_\Phi$ that acts on $n$ input qubits and at most $2c\log n +2$ ancilla qubits, and $U_\Phi$ consists of at most $34c^2 n^c \log^2 n + (70k+2)n^c + 35c\log n$ elementary gates such that $U_{\Phi}\ket{x}\reg{in}\otimes\ket{0}\reg{anc}=\ket{\Phi(x)}\reg{out}\otimes\ket{\psi_x}\reg{\overline{out}}$ for all $x\in\{0,1\}^n$, where $\ket{\psi_x}$ is some quantum state depending on $x$.
    The construction of $U_{\Phi}$ can be done in $\poly(n)$ time.\footnote{It is worth noting that an alternative quantum circuit for computing $k$CNF appears in~\cite{huang2020explicit}, Corollary 10. That construction has width $k + 2\bigl(\lceil \log n \rceil + \lceil c\log n \rceil\bigr)\text{ and size at most } 4 \times 3^{\lceil \log_2 n \rceil + \lceil c\log_2 n \rceil - 1}$. However, its circuit size is $O(n^{c+1})$, whereas ours is $O(n^{c}\log^2 n)$. This gap impacts the locality of the resulting Hamiltonian, which is why \cref{lem:U_compute_Phi} is necessary for our results.}
\end{lemma}

We defer the proof of \cref{lem:U_compute_Phi} to \cref{sec:SAT2C}.

\subsection{SETH- and QSETH-hardness of the 3-local Hamiltonian problem}\label{sec:Hlowerbound}
\begin{proof}[Proof of \cref{thm:SETH hardness for 3LH}]
    Let $\Phi=\varphi_1\wedge\varphi_2\wedge\cdots\wedge\varphi_m$ be a $k$CNF formula defined on $n$ variables $x_1,x_2,\dots,x_n$, where $m=O(n^c)$ and $c\in[1,2)$. 
    We construct an algorithm $\mathcal{A}_{\SAT}$ that solves $\kSAT(\Phi)$
    by using $\mathcal{A}_{\LH}$ as a subroutine. 
    
    \begin{enumerate}
        \item Upon receiving a $\kSAT$ instance $\Phi$, construct a quantum circuit $U_\Phi$ from \cref{lem:U_compute_Phi}.
        \item Construct a Hamiltonian $H_{U_\Phi}$ acting on $n_H=n+o(n)$ qubits from $U_\Phi$, and obtain two energy thresholds $\EnergyYes,\EnergyNo$ where $\EnergyNo-\EnergyYes = \Omega(1)$ by \cref{theorem:main}.
        \item Run $\mathcal{A}_{\LH}$ on the input $H_{U_\Phi}$. If $\lambda(H_{U_\Phi})\leq \EnergyYes$, then return YES.
         If $\lambda(H_{U_\Phi})\geq \EnergyNo$, then return NO.
    \end{enumerate}
    The running time of $\mathcal{A}_{\SAT}$ is $\poly(n)$. 

    We now show the correctness of $\mathcal{A}_{\SAT}$.
    By \cref{lem:U_compute_Phi}, we have the following: 
    \begin{itemize}
        \item If there exists an assignment $x\in\{0,1\}^n$ such that $\Phi(x)=1$, then $\Pr[1\gets U_\Phi(\ket{x})]=1$.
        \item If $\Phi(x)\neq 1$ for all assignments $x\in\{0,1\}^n$, then $\Pr[1\gets U_\Phi(\ket{x})]=0$ for all $x\in\{0,1\}^n$. 
    \end{itemize}
    Combining \cref{theorem:main} and choosing $\cerr < \frac{1}{4}$, we obtain:
    \begin{itemize}
        \item $||H_{U_\Phi}|| \leq \poly(n_H)$.
        \item If there exists $x\in\{0,1\}^n$ such that $\Phi(x)=1$, then $\lambda(H_{U_\Phi})\le \EnergyYes$.
        \item If $\Phi(x)\neq 1$ for all $x\in\{0,1\}^n$, then $\lambda(H_{U_\Phi})\ge \EnergyNo$.
        \item $\EnergyNo-\EnergyYes = \Omega(1)$
    \end{itemize}
    Therefore, $\mathcal{A}_{\SAT}$ has the same success probability as $\mathcal{A}_{\LH}$.

    Next, we show that $n_H = n + o(n)$ and that $H_{U_{\Phi}}$ is 3-local.
    Let $n_a$ and $g$ denote the number of ancilla qubits and elementary gates in $U_{\Phi}$, respectively.
    By \cref{lem:U_compute_Phi}, $n_a=o(n)$, and $g\le 34 c^2 n^c \log^2 n + (70k + 2)n^c + 35c \log n$. 
    For all sufficiently large $n$, the gate number $g$ is upper-bounded by $35 c^2 n^{c'}$, where $c' \in (c, 2)$ is a constant. 
     
    By \cref{theorem:main}, we can apply $d=2$ and  
    the number of qubits of $H_{U_\Phi}$ is $n_H=n+O(\log n)+o(n) = n +o(n)$
    and $H_{U_\Phi}$ is $3$-local.
    Furthermore, we have that the threshold gap $\EnergyNo-\EnergyYes = \Omega(1)$.
    
    Because $n_H=n+o(n)$, we have $n_H\le (1+\eta)n$ where $\eta>0$ is an arbitrary small constant for all sufficiently large $n$.\
    For any constant $\xi>0$, we can take $\xi>\eta>0$ and a constant $\varepsilon=\xi-\eta+\xi\eta$ such that $\varepsilon>0$, and it holds that $O(2^{n(1+\eta)(1-\xi)})\le O(2^{n(1-\varepsilon)})$ (and $O(2^{\frac{n}{2}(1+\eta)(1-\xi)})\le O(2^{\frac{n}{2}(1-\varepsilon)})$).  

    Assume there exists a classical (resp.  quantum) algorithm $\mathcal{A}_{\LH}$ that solves $3\textrm{-}\LH(H,\EnergyYes,\EnergyNo)$ for $\EnergyNo-\EnergyYes\ge \Omega(1)$ in $O(2^{n_H(1-\xi)})$ (resp. $O(2^{\frac{n_H}{2}(1-\xi)})$) time for a constant $\xi>0$. Then we have the  classical (resp. quantum) algorithm $\mathcal{A}_{\SAT}$ solves $\kSAT$ defined on $n$ variables in $O(2^{n(1-\varepsilon)})$ (resp. $O(2^{\frac{n}{2}(1-\varepsilon)})$) time for a constant $\varepsilon>0$. Besides, the reduction is polynomial time. As a result, it violates SETH which is defined in \cref{conjecture:SETH} (resp. QSETH which is defined in \cref{conjecture:QSETH}). By taking a contraposition, we have the statement.
\end{proof}

\subsection{Efficient verification circuit for SAT}
\label{sec:SAT2C}
In this section, we show the construction of the quantum circuit $U_\Phi$ that computes the formula $\Phi=\varphi_1\wedge\dots\wedge\varphi_m$ defined on $n$ variables $x_1,x_2,\dots x_n$, and $m=O(n^c)$.
For each $i\in[m]$, the clause $\varphi_i=\ell_{i,1}\vee \ell_{i,2}\vee\cdots\vee\ell_{i,k_i}$ where for each $p\in[k_i]$, $\ell_{i,p}$ can be either $x_j$ or $\neg x_j$ where $j\in [n]$, and $k_i\le k$.
Let $r:=\ceil{\log m}$.

\begin{proof}[Proof of \cref{lem:U_compute_Phi}]
    The ideal is to compute $\varphi_i$ for each $i\in[m]$ sequentially.
    We initially set a counter, and increment the counter by one if $\varphi_i(x)=1$.
    After compute all $\varphi_i$, we check whether the counter is equal to $m$.
    Let $k_i$ be the number of variables ($x_j$ or $\neg x_j$) appearing in $\varphi_i$.
    We have $\Phi(x)=1$ if and only if the counter equals $m$.
    We introduce the ancilla register $\rgst{anc}=\rgst{cls}\cup \rgst{cnt} \cup \rgst{out}$ where $\rgst{cls}$ is a one-qubit register that temporarily stores the value of $\varphi_i(x)$ for each $i$, $\rgst{cnt}$ is an $r$-qubit ($r=\ceil{\log m}$) that serves as the counter, and $\rgst{out}$ is a one-qubit register that outputs $\Phi(x)$.

    To compute each $\varphi_i$, we construct $W_i$ acting on $\rgst{in}\cup \rgst{cls}$ such that $W_i\ket{x}\reg{in}\ket{0}\reg{cls}=\ket{\psi_x}\reg{in}\ket{\varphi_i}\reg{cls}$, where $\ket{\psi_x}$ is some state depending on $x$. 

    To increment the counter by one, we construct a unitary $\addone$ acting on $\rgst{cnt}$ such that $\addone\ket{y}\reg{cnt}=\ket{bin(int(y)+1)}\reg{cnt}$ for all $y\in\{0,1\}^r\setminus\{1^r\}$.
    We apply $\addone$ only when $\varphi_i(x)=1$. 
    This is done by letting $\addone$ be controlled by the resister $\rgst{cls}$.
    Denote the control-$\addone$ operator by $C\addone$.
    Then $C\addone\reg{cls\cup cnt}\ket{0}\reg{cls}\ket{y}\reg{cnt}=\ket{0}\reg{cls}\ket{y}\reg{cnt}$ and $C\addone\reg{cls\cup cnt}\ket{1}\reg{cls}\ket{y}\reg{cnt}=\ket{0}\reg{cls}\ket{bin(int(y)+1)}\reg{cnt}$.

    After calculating $\varphi_i(x)$ and applying $C\addone$, we apply $W_i^\dagger$ on $\rgst{in}\cup \rgst{cls}$ to restore the state to $\ket{x}\reg{in}\ket{0}\reg{cls}$. 

    To check whether the counter is equal to $m$ in the final step, we construct a compare operator, denoted by $\compare$, that acts on $\rgst{cnt}\cup\rgst{out}$.
    The operator $\compare$ satisfies $\compare \ket{y}\reg{cnt}\ket{0}\reg{out}=\ket{y}\reg{cnt}\ket{\delta_{m,int(y)}}\reg{out}$ for all $y\in\{0,1\}^{r}$. 

    To sum up, we construct the quantum circuit $U_{\Phi}$ as follows: 
    \begin{equation}
        U_\Phi:= \compare(W_m^\dagger C\addone W_m) (W_{m-1}^\dagger C\addone W_{m-1})\cdots(W_1^\dagger C\addone W_1).
    \end{equation}
    
    %=====================================================
    Next, we explain how to implement $W_i$, $\addone$, and $\compare$ gates.

    Because $\ell_{i_1}\vee \ell_{i_2}\vee\cdots\vee \ell_{i_{k_i}} = \neg(\neg \ell_{i_1} \wedge \neg \ell_{i_2}\wedge \cdots \wedge \neg \ell_{i_{k_i}})$, we can use $C^k\notgate$ together with $\notgate$ gates to implement $W_i$.
    Let $\rgst{S_i}\subseteq\rgst{in}$ be the register defined by $\rgst{S_i}:=\{j\in[n]:x_j\in\varphi_i\mathrm{\;or\;}\neg x_j \in\varphi_i \}$, i.e., $\rgst{S_i}$ consists of the $j$th qubits in $\rgst{in}$ such that $x_j$ or $\neg x_j$ in $\varphi_i$;
    and let $\rgst{R_i}\subseteq\rgst{in}$ be the register defined by $\rgst{R_i}:=\{j\in[n]:x_j\in\varphi_i\}$, i.e., $\rgst{R_i}$ consists of the $j$th qubits in $\rgst{in}$ such that $x_j$ in $\varphi_i$, but $\neg x_j$ does not. 
    We construct $W_i$ by 
    $W_i:=\notgate\cdot\reg{cls}C^{k_i}\notgate\reg{S_i\cup cls}\cdot\notgate^{\otimes |\rgst{R_i}|}\reg{R_i}$.

    To implement the $\addone$, observe that for $y\in\{0,1\}^r$, when $int(y)$ is incremented by one, a bit $y[p]$ will be flipped if and only if all the bits with order lower than $p$ are 1's.
    That is, $y[p]$ will be flipped if and only if for all $p'>p$, $y[p'] = 1$.
    Hence, $\addone$ can be composed of a sequence of multi-control-NOT gates.
    The $q$th layer of $\addone$ is a $C^{r-q}\notgate$ gate whose control qubits are ${q,q+1,\dots,r}$ and whose target is the $q$the qubit.
    The last layer of $\addone$ is a $\notgate$ gate acting on the last qubit.
    To implement the control-$\addone$ operation $C\addone$, we let every $C^{r-q}\notgate$ be controlled by register $\rgst{cls}$.
    In other words, the  $q$th layer of $C\addone$ is a $C^{r+1-q}\notgate$ gates whose control qubits are ${q,q+1,\dots,r}\cup\rgst{cls}$ and whose target is the $q$th qubit in $\rgst{cnt}$.

    The $\compare$ operator checks whether the value stored in the counter equal to $m$.
    That is, $\compare$ operator compares each bit in the counter with $bin(m)$, which can be implemented by a $C^r\notgate$ gate. 
    Let $\rgst{P}\subseteq\rgst{cnt}$ be defined by $\rgst{P}:=\{j\in[r]:bin(m)[j]=0\}$.
    We construct $\compare$ by $\compare:=C^r\notgate\reg{cnt\cup out}\cdot \notgate^{\otimes |P|}\reg{P}$.  
    
    %===========================================================
    We decompose each multi-control-NOT gate into Toffoli gates.
    For $W_i$ (and $W_i^\dagger$), there is a $C^{k_i}\notgate$ inside, and it can be decomposed into at most $2k$ of Toffoli gates using at most $k$ ancilla bits.

    The largest gate in the $C\addone$ is a control $C^r\notgate$ gate, and there are $r$ layers.
    Hence, $C\addone$ can be decomposed into at most $2r^2$ Toffoli gates.
    Because the ancillas for the multi-control-NOT gates can be reused, the ancillas required for $C\addone$ is at most $r$.

    The $\compare$ operation contains a $C^{r}\notgate$ gate, which can be decomposed into at most $2r$ Toffoli gates using at most $r$ ancilla qubits.

    The ancilla qubits can be reused for the multi-control-NOT gates. 
    Hence, to decomposed $U_\Phi$ into Toffoli gates, the number of  ancilla qubits required is at most $r$. 
    Therefore, the total number of ancilla qubits for $U_\Phi$ is $|\rgst{anc}|+r = |\rgst{cls}|+|\rgst{cnt}|+|\rgst{out}|+r = 2\log m +2 = 2c\log n +2$.

    The total number of Toffoli gates in $U_\Phi$ is at most $(4k+2r^2)\cdot m + 2r = 2c^2 n^c \log^2 n + 4k n^c +2c\log n$.
    We further decompose the Toffoli gates into elementary gates,
    yielding $34c^2 n^c \log^2 n + 68kn^c +34c\log n$ elementary gates in total.
    In addition to Toffoli gates, there are at most $2(k+1)\cdot m + r= (2k+2)n^c +  c\log n$ of $\notgate$ gates.
    Therefore, the total number of gates is at most $34c^2 n^c c^2 \log^2 n + (70k+2)n^c +35c \log n$.
\end{proof}
\section{Fine-grained complexity of quantum partition function problem}\label{sec:fine-grained complexity of QPF new}

The fine-grain lower bound for the local Hamiltonian problem directly gives us the fine-grain lower bound for the quantum partition function.

\subsection{SETH- and QSETH-hardness of quantum partition functions}\label{subsec:seth-hardness of QPF}

\begin{lemma}[Fine-grained reduction from $\kLH$ to $\kQPF$]
    \label{lem:FG_reduction_of_QPF}
    For any $T(n)\in\omega(\poly(n))$, $\kLH(H,\EnergyYes,\EnergyNo)$ is $(T(n), T(n))$ reducible to $\kQPF(H,\beta, \delta)$ in which $H$ acts on $n$ qubits, $\beta\geq \frac{n}{\EnergyNo-\EnergyYes}$ and $\delta$ satisfies that $\frac{1-\delta}{1+\delta}\geq e^{-0.3n}$.
\end{lemma}
We emphasize that in \cref{lem:FG_reduction_of_QPF},  $H$ in the $\kQPF$ problem and the $\kLH$ problem is the same Hamiltonian.

    To better understand \cref{lem:FG_reduction_of_QPF}, we unpack the underlying fine-grained reduction as follows. 
    Suppose there exist $\xi>0$ and an algorithm $\mathcal{A}_{\QPF}$ that approximate $\tr[e^{-\beta H}]$ up to relative error $\delta$ satisfying $\frac{1-\delta}{1+\delta}\ge e^{-0.3n}$ with probability greater than $2/3$ in $O(T(n)^{1-\xi})$ time, 
    then for any $\varepsilon>0$, we can use $\mathcal{A}_{\QPF}$ to construct an algorithm $\mathcal{A}_{\LH}$ such that given thresholds $\EnergyYes,\EnergyNo$ satisfying $\EnergyNo-\EnergyYes\ge n/\beta$, the algorithm $\mathcal{A}_{\LH}$ decides weather $\lambda(H)\ge \EnergyNo$ or $\lambda(H)\le \EnergyYes$ with probability greater than $2/3$ in $O(T(n)^{1-\varepsilon})$ time.  

    On the other hand,
    suppose for all $\varepsilon>0$, for any algorithm $\mathcal{A}_{\LH}$, for infinitely many $n$ there exists a $k$-local Hamiltonian $H_n$ acting on $n$ qubits and two energy thresholds $\EnergyYes,\EnergyNo$ such that $\mathcal{A}_{\LH}$ cannot solve $\kLH(H_n,\EnergyYes,\EnergyNo)$ with probability greater than $\frac{2}{3}$ in $O(T(n)^{1-\varepsilon})$ time,    
    then for all $\xi>0$, for all $\delta>0$, for infinitely many $n\ge n_1$, 
    any algorithm $\mathcal{A}_{\QPF}$ cannot solve $\kQPF(H_n, \beta, \delta)$, where $\beta\ge\frac{n}{\EnergyNo-\EnergyYes}$, with probability greater than $2/3$ in $O(T(n)^{1-\xi})$ time.
    The integer $n_1$ satisfies that $\frac{1-\delta}{1+\delta}=e^{-0.3n_1}$.

\begin{proof}[Proof of \cref{lem:FG_reduction_of_QPF}]
    Let $(H,\EnergyYes,\EnergyNo)$ be the $\kLH$ instance, where $H$ is a $k$-local Hamiltonian acts on $n$ qubits .
    We are going to construct an algorithm $\mathcal{A}_{\LH}$ that solves $\kLH(H, \EnergyYes, \EnergyNo)$ by using $\mathcal{A}_{\QPF}$ as a subroutine.
    
    If $(H, \EnergyYes, \EnergyNo)$ is YES case, then there is at least one eigenstate of $H$ whose energy is lower than or equal to $\EnergyYes$.
    Hence $Z(\beta)\ge e^{-\beta \EnergyYes}$.
    
    If $(H, \EnergyYes, \EnergyNo)$ is NO case, then all $2^n$ number of eigenstates of $H$ have energy higher than or equal to $\EnergyNo$.
    Hence, $Z(\beta)\le 2^n e^{-\beta \EnergyNo} < e^{-\beta \EnergyNo +0.7n}$.
    
    When $\beta \ge \frac{n}{\EnergyNo-\EnergyYes} \in \poly(n)$, and $\delta$ satisfies $\frac{1-\delta}{1+\delta} \ge e^{-0.3n}$, it holds that 
    \[
    \frac{1-\delta}{1+\delta}\ge e^{-0.3n}\ge e^{-\beta(\EnergyNo-\EnergyYes)+0.7n}.
    \] 
    Hence we have $(1-\delta)e^{-\beta \EnergyYes}\ge (1+\delta)e^{-\beta \EnergyNo +0.7n}$.
    
    Now we present the algorithm $\mathcal{A}_{\LH}$ that solves $\kLH(H, \EnergyYes, \EnergyNo)$ by using $\mathcal{A}_{\QPF}$.
    \begin{enumerate}
        \item When receiving $H, \EnergyYes, \EnergyNo$, calculate $\delta_0$ such that $\frac{1-\delta_0}{1+\delta_0} = e^{-0.3n}$ and $\beta_0=\frac{n}{\EnergyNo-\EnergyYes}$.
        Set $\beta \ge\beta_0$ and set $\delta \le\delta_0$.
        \item Run $\mathcal{A}_{\QPF}$ on the input $H, \beta, \delta$, and get the output $\zout$.
        \item If $\zout\ge (1-\delta)e^{-\beta \EnergyYes}$, then output YES.
        If $\zout\le (1+\delta)e^{-\beta \EnergyNo +0.7n}$, then output NO.
    \end{enumerate}
    When $\mathcal{A}_{\QPF}$ solves $\kQPF(H, \beta, \delta)$ successfully, it is guaranteed that $(1-\delta)Z(\beta) \le \zout \le (1+\delta)Z(\beta)$.
    When the inputs of $\kLH$ is YES case, $\zout\ge (1-\delta)Z(\beta) \ge (1-\delta)e^{-\beta \EnergyYes}$; and when NO case, $\zout\le (1+\delta)Z(\beta) < (1+\delta)e^{-\beta \EnergyNo +0.7n}$.

    By the choice of $\beta$ and $\delta$, it holds that $(1-\delta)e^{-\beta \EnergyYes}\ge (1+\delta)e^{-\beta \EnergyNo +0.7n}$.
    Therefore, when $\mathcal{A}_{\QPF}$ solves $\kQPF(H, \beta, \delta)$ successfully, $\mathcal{A}_{\LH}$ decides $\kLH(H, \EnergyYes, \EnergyNo)$ correctly.
    The success probability of $\mathcal{A}_{\LH}$ is the same as $\mathcal{A}_{\QPF}$.
    Therefore, the requirement \ref{requre1} in \cref{def:fine-grain_reduction} is satisfied.

    The algorithm $\mathcal{A}_{\LH}$ queries $\mathcal{A}_{\QPF}$ once in Step 2., and the instance of $\kQPF$ is exactly the input of $\kLH$.
    For any $\varepsilon$, we choose $\xi=\varepsilon$. 
    We have $T(n)^{1-\varepsilon}\le d\cdot T(n)^{1-\xi}$ for any constant $d>1$.
    Therefore, the requirement \ref{requre3} and \ref{requre4} in \cref{def:fine-grain_reduction} are satisfied.

    Finally, Step 1.\ and Step 3.\ in the algorithm $\mathcal{A}_{\LH}$ run in $\poly(n)$ time.
    Hence, $\mathcal{A}_{\LH}$ runs in $\poly(n)$ times, which is less than $d\cdot T(n)^{1-\xi}$ for some constant $d$ because $T(n)$ is super-polynomial.
    consequently,  the requirement \ref{requre2} is satisfied.
    This finishes the proof.
\end{proof}

\begin{remark}
    The reduction in \cref{lem:FG_reduction_of_QPF} is polynomial-time. Since $k\textrm{-}\LH$ is $\QMA$-hard, we have that $\kQPF$ is also $\QMA$-hard.
\end{remark}

The fine-grain reduction above implies the lower bound for the quantum partition function problem.

\begin{theorem}
    \label{thm:main_qpf}
    Assume SETH (resp. QSETH) holds.
    For any $\xi>0$ and any $\delta\in(0,1)$, 
    for any quantum algorithm $\mathcal{A}_{\QPF}$, for infinitely many $n$,
    there exists a $3$-local Hamiltonian $H$ acting on $n$ qubits, associated with an inverse temperature $\beta_0=O(n)$, such that for all $\beta\ge \beta_0$, the algorithm $\mathcal{A}_{\QPF}$ cannot solve $\QPF(H,\beta, \delta)$ in $O(2^{n(1-\xi)})$ classical time (resp. $O(2^{\frac{n}{2}(1-\xi)})$ quantum time) with probability greater than $2/3$.
\end{theorem}

\begin{proof}
    We prove the theorem by contradiction.
    Assume that there exists an algorithm $\mathcal{A}_{\QPF}$ that solves the $\QPF$ problem.
    Then, by \cref{lem:FG_reduction_of_QPF}, we can construct an algorithm $\mathcal{A}_{\LH}$ that solves 3-local Hamiltonian problem, which contradicts \Cref{thm:SETH hardness for 3LH}.  
     
    Suppose there exist $\xi>0$ and $\delta\in(0,1)$ such that there exist an algorithm $\mathcal{A}_{\QPF}$ and $n_1\in\mathbb{N}$ satisfying the following: for all $n>n_1$, for all $H$ acting on $n$ qubits, and for arbitrarily large $\beta\ge O(n)$, the algorithm $\mathcal{A}_{\QPF}$ solves $\QPF(H,\beta, \delta)$ with probability greater than $2/3$ in $O(2^{n(1-\xi)})$ classical time.
    
    Let $(H,\EnergyYes,\EnergyNo)$ be an LH instance, where $H$ is a $3$-local Hamiltonian acting on $n$ qubits with sufficiently large $n$ (we will specify how large $n$ shall be later), and $\EnergyNo-\EnergyYes=\Omega(1)$.
    Choose $\beta$ such that $\beta \ge O(n)>\frac{n}{\EnergyNo-\EnergyYes}$.
    By \cref{lem:FG_reduction_of_QPF}, we construct $\mathcal{A}_{\LH}$, in which we query $\mathcal{A}_{\QPF}$ on the instance $(H, \beta, \delta)$.
    The algorithm runs in $O(2^{n(1-\xi)})$ time.
    
    Let $n_2$ be the smallest integer such that $\frac{1-\delta}{1+\delta}\ge e^{-0.3n_2}$.
    By the hypothesis at the beginning of the proof and by \cref{lem:FG_reduction_of_QPF}, when $n\ge\max\{n_1,n_2\}$, the algorithm $\mathcal{A}_{\LH}$ decides $(H,\EnergyYes,\EnergyNo)$ successfully with probability greater than $2/3$.

    Note that the reduction works for all $H$ acting on $n$ qubits and $\EnergyYes,\EnergyNo$ satisfying $\EnergyNo-\EnergyYes=\Omega(1)$ as long as $n\ge\max\{n_1,n_2\}$.
    We thus obtain a contradiction with \Cref{thm:SETH hardness for 3LH}. The lower bound of quantum algorithms assuming QSETH follows in the same proof strategy.
    % This completes the proof.
\end{proof}

\subsection{$O^{\ast}(2^{\frac{n}{2}})$-time quantum algorithm for $\kQPF$ problem}

In this section, we propose an algorithm that estimates the quantum partition function under a polynomial inverse temperature up to an arbitrary inverse polynomial error in $O^{\ast}(2^{\frac{n}{2}})$ time.
The running time matches our conditional lower bound.
The input Hamiltonian $H$ is $n$-qubit ,constant-local, semi-positive and satisfies $\|H\|< 1$.

\begin{theorem}[$O^{\ast}(2^{\frac{n}{2}})$ time algorithm for $\kQPF$]
    \label{thm:main_qpf_alg}
    For all sufficiently large $n\in\mathbb{N}$, for all $n$-qubit constant-local, semi-definite Hamiltonian $H$ satisfying $\|H\|<1$,  $\beta<\poly(n)$, and $\delta\ge{1}/{n^c}$ for arbitrary positive constant $c$,
    there exists a quantum algorithm that solves $\kQPF(H, \beta, \delta)$ in $O^{\ast}(2^{\frac{n}{2}})$ time with successful probability $1-\negl(n)$.  
\end{theorem}

\subsubsection*{Tools that are used in the proof of QPF algorithm}

Before proving \cref{thm:main_qpf_alg}, we first list some tools that will be used in the proof.

The main idea of the algorithm is dividing the energy range $[0,1)$ into polynomial many disjoint intervals $I_0, I_1, \dots, I_L$.
For each interval $I_\ell$, we designate a value $\mathcal{E}_\ell\in I_\ell$ to represent the characteristic energy of $I_\ell$. 
Let $M_\ell$ be the number of the Hamiltonian's eigenstates whose corresponding energy is in the interval $I_\ell$.
The approximation of the quantum partition function is 
\begin{equation*}
    \sum_{\ell=0}^{L} M_\ell e^{-\beta \mathcal{E}_\ell}.
\end{equation*}
To estimate the number of eigenstates $M_\ell$, there are two key ingredients: (i) energy estimation, and (ii) quantum counting.

Given the description of a local Hamiltonian $H$ and one of $H$'s eigenstate $\ket{\psi_j}$, 
the energy estimation algorithm outputs its corresponding energy $E_j$.
The energy estimation algorithm can be implemented by phase estimation associated with Hamiltonian simulation. 

\begin{lemma}[Phase estimation \cite{NC10}]\label{lem:phase_estimation}
    Let $U$ be a $n$-qubit unitary and $\ket{u_\theta}$ be $U$'s eigenstate satisfying that $U\ket{u_\theta}=e^{i\theta}\ket{u_\theta}$.
    There exists a quantum circuit $P_{U,r}$, which takes $\ket{u_\theta}$ as the input state and executes control-$U$ operator $O(2^r)$ times with additional  $\poly(n)$ number of gates, such that the output state is 
    \begin{equation}
    \label{eq:PE_out}
        \ket{u_\theta}\otimes \sum_{x\in\{0,1\}^r}\alpha_{x}\ket{x},
    \end{equation}
    where $\alpha\in\mathbb{C}$, satisfying that 
    \begin{equation}
    \label{eq:PE_result}
        \sum_{x:|\frac{int(x)\cdot2\pi}{2^r}-\theta|\le \frac{2\pi}{2^b}} |\alpha_x|^2\ge 1- \frac{1}{2^{r-b}}
    \end{equation}
    for any $b<r-1$.
\end{lemma}

When we choose $b=r-2$, we have that the estimation $\widetilde{\theta}:=\frac{int(x)\cdot 2\pi}{2^r}$ is in the interval $\theta\pm \frac{2\pi}{2^b}$ with probability greater than $\frac{3}{4}$.
We can amplify the probability to $1-\negl(n)$ by repeating the procedure $\poly(n)$ times and using median of means technique \cite{AA17, Bravyi+19}.

\begin{lemma}[Confidence amplification of phase estimation]\label{lem:amplify_pe}
    Following \cref{lem:phase_estimation}, there exists a quantum circuit $A_{U,r,m}$, which executes $P_{U,r}$ $m=\poly(n)$ times, with additional $\poly(n)$ gates, such that the output state is
    \begin{equation}
    \label{eq:APE_output}
        \ket{u_\theta}\otimes \sum_{x\in\{0,1\}^r}\alpha'_{x}\ket{x}\ket{g_x},
    \end{equation}
    where $\alpha'\in\mathbb{C}$, and $\ket{g_x}$ is some state depending on $x$, satisfying
    \begin{equation}
    \label{eq:APE_result}
        \sum_{x:|\frac{int(x)\cdot2\pi}{2^r}-\theta|\le \frac{2\pi}{2^{(r-2)}}} |\alpha'_x|^2\ge 1- e^{-\Theta(m)}.
    \end{equation}
\end{lemma}
\begin{proof}
    Because after the  execution of the phase estimation $P_{U,r}$, the input state $\ket{u_\theta}$ remains unchanged and separable from $\sum_{x}\alpha_x\ket{x}$, we can reuse the input state.
    Assume we measure the register that contains $\sum_{x}\alpha_x\ket{x}$ after each execution of $P_{U,r}$. Let the outcome of each execution be $x_1$, $x_2, \dots, x_m$.
    By \cref{lem:phase_estimation}, when we choose $b=r-2$, it holds that for each $i\in[m]$ the probability that $\widetilde{\theta_i}:=\frac{int(x_i)\cdot 2\pi}{2^r}$ is in the interval $\theta\pm\frac{2\pi}{2^{r-2}}$ is greater than $\frac{3}{4}$.

    Because the register that stores  $\sum_{x}\alpha_x\ket{x}$ for each execution are separable, we have that the measurement outcomes of all the execution are independent.  
    Then, by Chernoff bound, the probability that there are more than half of $\widetilde{\theta_1}$, $\widetilde{\theta_2},\dots, \widetilde{\theta_m}$ outside the interval $\theta\pm \frac{2\pi}{2^{r-2}}$ is less than $e^{-\Theta(m)}$.
    Let $\widetilde{x}$ be the the median of $int(x_1)$, $int(x_2),\dots,int(x_m)$ and let $\widetilde{\theta}:=\frac{\widetilde{x}\cdot 2\pi}{2^r}$. 
    We have the event that $\widetilde{\theta}$ is outside the interval $\theta\pm \frac{2\pi}{2^{r-2}}$ has probability less than $e^{-\Theta(m)}$.

    Now we describe the quantum circuit $A_{U,r,m}$.
    The quantum circuit $A_{U,r,m}$ has $m$ ancilla registers $\rgst{X}_1$, $\rgst{X}_2, \dots, \rgst{X}_m$ each of them is $r$-qubit, and an input register $\rgst{U}$ takes the input state $\ket{u_\theta}$.
    The quantum circuit $A_{U,r,m}$ executes $P_{U,r}$ $m$ times sequentially on the input $\ket{u_\theta}$ and in the $i$th execution of $P_{U,r}$, it stores $\sum_x\alpha_x\ket{x}$ in the register $\rgst{X}_i$.
    After $m$ times execution, a circuit that computes the median of $\rgst{X}_1$, $\rgst{X}_2, \dots, \rgst{X}_m$ is applied. Since $\rgst{X}_1$, $\rgst{X}_2, \dots, \rgst{X}_m$ are separable and computing the median is a classical circuit, the distribution of $\widetilde{x}$ is the same regardless of the measurement taking place before or after the computing of median.
    As a result, Equation \cref{eq:APE_result} holds.
\end{proof}

To implement the energy estimation of the local Hamiltonian $H$, we run phase estimation for $U=e^{-iH}$, which can be implemented by Hamiltonian simulation.  
\begin{lemma}[Hamiltonian simulation \cite{LC17}]\label{lem:Hsim}
    For all sufficiently large $n\in\mathbb{N}$, for all constant-local Hamiltonian $H\in\mathbb{C}^{2^c\times 2^n}$,  evolution time $t\in \mathbb{R}^{+}$, and tolerated error $\varepsilon\in(0,1)$, there exists a quantum circuit $U_{\operatorname{sim}}(H,t,\varepsilon)$ running in $\poly(n,t,\log({1}/{\varepsilon}))$ such that for any quantum state $\ket{\psi}\in\mathbb{C}^{2^n}$, it holds that
    \begin{equation}\label{eq:sim}
        \|(U_{\operatorname{sim}}(H,t,\varepsilon)-e^{-iHt}\otimes I)\ket{\psi}\otimes\ket{0^{a_{\operatorname{sim}}}}\|_2\le\varepsilon,
    \end{equation}
    where $a_{\operatorname{sim}}$ is the number of ancilla qubits of $U_{\operatorname{sim}}(H,t,\varepsilon)$.
\end{lemma}

Now we formally define the energy estimation.
\begin{definition}[Energy estimation]\label{def:energy_estimation}
    Let $H\in\mathbb{C}^{2^n\times 2^n}$ be a constant-local Hamiltonian satisfying that all the eigenvalues of $H$ are in the interval $[0,1)$. 
    We first define an ideal unitary $\widehat{U}_{\operatorname{EE}}(H, \delta E, \eta)$ that estimates the corresponding energy of $H$'s eigenstate $\ket{\psi_j}$.
    We say $\widehat{U}_{\operatorname{EE}}(H, \delta E, \eta)$ is an ideal energy estimation of the Hamiltonian $H$ if 
    for any $H$'s eigenstate $\ket{\psi_j}$, it holds that
    \begin{equation}\label{eq:energy_estimation}
        \widehat{U}_{\operatorname{EE}}(H,\delta E, \eta)\ket{\psi_j}\ket{0^{a_{\operatorname{EE}}}} 
        = \ket{\psi_j}\otimes\sum_{E'} \alpha_{E'}\ket{E'}\ket{g_{E'}},
    \end{equation}
    where $a_{\operatorname{EE}}$ is the number of ancilla qubits, $\alpha_{E'}\in\mathbb{C}$, and $\ket{g_{E'}}$ is some state depending on $E'$, satisfying
    \begin{equation}
        \label{eq:ee_require}
        \sum_{E'\in[E_j-\delta E, E_j+\delta E]}|\alpha_{E'}|^2 \ge 1-\eta,
    \end{equation}
    where $E_j$ is the corresponding energy of $\ket{\psi_j}$.

    Then, we say a quantum circuit $U_{\operatorname{EE}}(H, \delta E, \eta, \varepsilon)$ is an implementation of energy estimation of $H$ if
    \begin{equation}
        \label{eq:EE_implementation}
        \|U_{\operatorname{EE}}(H, \delta E, \eta, \varepsilon) - \widehat{U}_{\operatorname{EE}}(H,\delta E, \eta)\|\le \varepsilon.
    \end{equation}
\end{definition}

\begin{lemma}[Efficient energy estimation algorithm]\label{lem:energy_estimation}
    For all sufficiently large $n\in\mathbb{N}$, 
    For all constant-local, semi-definite Hamiltonian $H\in\mathbb{C}^{2^n\times 2^n}$ satisfying $\|H\|<1$, and all tolerated error $\delta E = 1/\poly(n)$,
    there exists a energy estimation implementation $U_{\operatorname{EE}}(H,\delta E,\eta=e^{-n}, \varepsilon=2^{-n})$ that runs in $\poly(n)$ time.
\end{lemma}

\begin{proof}
    It holds that $e^{-iH}\ket{\psi_j}=e^{-iE_j}\ket{\psi_j}$.
    We can use phase estimation to obtain $E_j$.
    We first consider the ideal unitary $\widehat{U}_{\operatorname{EE}}(H, \delta E, \eta)$.
    By \cref{lem:amplify_pe}, we run $A_{U,r,m}$ with  $U=e^{-iH}$, $r:=\log(\frac{2\pi}{\delta E})+2$, and $m=Cn$ where $C$ is a large enough constant, on the input $\ket{\psi_j}$. 
    Let $E':={int(x)\cdot2\pi}/{2^r}$,  we have that Equation \cref{eq:energy_estimation} and \cref{eq:ee_require} holds with $\eta=e^{-n}$.

    Then, we consider the efficient implementation $U_{\operatorname{EE}}(H,\delta E, \eta=e^{-n} ,\varepsilon=2^{-n})$.
    We replace $U$ in $A_{U,r,m}$  with the implementation of Hamiltonian simulation $U_{\operatorname{sim}}(H, t=1, \varepsilon/(m\cdot 2^r))$ described in \cref{lem:Hsim}. 
    (To be more precise, we need to implement control-$U$ in $A_{U,r,m}$. Hence, we need to replace the gates in $U_{\operatorname{sim}}(H, t=1, \varepsilon/(m\cdot2^r))$ with their controlled version.)
    There are $m\cdot 2^{r}$ many of $U$ in $A_{U,r,m}$.
    By triangular inequality, Equation \cref{eq:EE_implementation} holds.

    For the running time, by \cref{lem:Hsim}, the implementation of $U_{\operatorname{sim}}(H,t=1,\varepsilon/(m\cdot 2^r))$ takes $\poly(n, 1, \log(/m\cdot 2^r\varepsilon)) = \poly(n, 1 ,(\log(n)+\delta E+\log(1/\varepsilon)))$ time.
    We have $\delta E = 1/\poly(n)$ and $\varepsilon=2^{-n}$, the running time of each implementation of $U$ is $\poly(n)$.
    There are $m\cdot 2^{r}=\poly(n)$ many of $U$ in $A_{U,r,m}$.
    As a result the overall running time of $U_{\operatorname{EE}}(H,\delta E, \eta=e^{-n} ,\varepsilon=2^{-n})$ is $\poly(n)$.
\end{proof}

We use quantum counting to estimate the number of eigenstates in a given interval.
\begin{lemma}[Quantum counting]\label{lem:amplitude_estimation}
    Let a quantum circuit $U$ taking a $n$-qubit input state $\ket{\psi}$ satisfy $U\ket{\psi}\ket{0^a}
    =\sqrt{\frac{M}{N}}\ket{1}\ket{\xi^1} 
    + \sqrt{\frac{N-M}{N}}\ket{0}\ket{\xi^0}$,
    where $N=2^n$, $M\le N$\footnote{Actually it requires that $M<N/2$. We can overcome this issue by adding an additional qubit.}, and $a$ is the number of ancilla qubits.
    For any $c\in\mathbb{N}$, there is a quantum circuit that executes $U$ $O(\poly(n)\cdot 2^{\frac{n}{2}})$ times with additional $O(\poly(n)\cdot 2^{\frac{n}{2}})$ many gates and outputs $\widetilde{M}$ such that $(1-\Theta(\frac{1}{n^c}))M\le\widetilde{M}\le(1+\Theta(\frac{1}{n^c}))M$ with probability $1-\negl(n)$.
\end{lemma}

In \cite{NC10}, there has already been a proof for the amplitude estimation algorithm that outputs $\widetilde{M}$ satisfying $(1-\frac{1}{\sqrt{M}})M\le\widetilde{M}\le (1+\frac{1}{\sqrt{M}})M$.
To achieve the relative error as small as $\Theta(1/n^c)$, we increase the number of Grover iteration in \cite{NC10} from $n/2+3$ to $n/2+\Theta(c\log n)$.

To  use quantum counting to estimate the number of eigenstates in a given interval, we need to prepare the uniform superposition of all the eigenstates as the input state.
However, in general, the eigenstates are unknown.
We use the property that EPR state is identical to the uniform superposition of the states and its complex conjugate in any orthogonal basis to overcome this issue.
\begin{lemma}\label{lem:choi}
    For any complete orthogonal basis of $n$-qubit system $\{\ket{\psi_i}\}_{j=1}^{N}$,i.e., $\sum_{j=1}^{N}\proj{\psi_j}=I$ and $\inner{\psi_j}{\psi_{j'}}=\delta_{j,j'}$ where $N=2^n$, it holds that
    \begin{equation}
        \sum_{j=1}^{N}\sqrt{\frac{1}{N}}\ket{\psi_j}\ket{\psi_j^{\ast}} = \sum_{k=1}^{N}\sqrt{\frac{1}{N}}\ket{k}\ket{k}.
    \end{equation}
    where $\ket{\psi_j^\ast}$ is the complex conjugate of $\ket{\psi_j}$.
\end{lemma}
\begin{proof}
    Let $V=\sum_{j,k=1}^{N}V_{k,j}\ketbra{k}{j}$ be a unitary such that $\ket{\psi_j}= V\ket{j}$ for all $j\in[N]$.
    We have 
    \begin{align}
        \sum_{j=1}^{N}\ket{\psi_j}\otimes\ket{\psi_j^{\ast}} 
        & = \sum_{j=1}^{N}\bigg(\Big(\sum_{k=1}^{N}V_{k,j}\ket{k}\Big)\otimes\Big(\sum_{k'=1}^{N}V_{k',j}^{\ast}\ket{k'}\Big)\bigg)\nonumber\\
        & =  \sum_{j=1}^{N}\sum_{k=1}^{N}\sum_{k'=1}^{N}V_{k,j}V^{\ast}_{k',j}\ket{k}\otimes\ket{k'}.\label{eq:Choi_1}
    \end{align}
    Because $V$ is a unitary, we have $\sum_{j=1}^{N}V_{k,j}V^{\ast}_{k',j}=\delta_{k,k'}$.
    Equation~(\ref{eq:Choi_1}) becomes $\sum_{j=1}^{N}\ket{k}\otimes\ket{k}$.
\end{proof}

\subsubsection*{Proof of \cref{thm:main_qpf_alg}}
\renewcommand{\labelenumii}{\arabic{enumi}.\arabic{enumii}}
\renewcommand{\labelenumiii}{\arabic{enumi}.\arabic{enumii}.\arabic{enumiii}}
\newcommand{\dcoun}{\delta_{\operatorname{C}}}

\begin{proof}[Proof of \cref{thm:main_qpf_alg}]
    Let $N:=2^n$.

    For any interval $I$ and $E\in[0,1)$, we define the operator $U_{\operatorname{dec}}(I)\ket{E}\ket{0^a_{\operatorname{dec}}}=\ket{\indicator_{I}(E)}\ket{\ast}$, where $a_{\operatorname{dec}}$ is the number of ancilla qubits and $\ket{\ast}$ is some state. 
    That is, $U_{\operatorname{dec}}(I)\ket{E}\ket{0^a_{\operatorname{dec}}}=\ket{1}\ket{\xi_{E,1}}$ if $E\in I$ and $U_{\operatorname{dec}}(I)\ket{E}\ket{0^a_{\operatorname{dec}}}=\ket{0}\ket{\xi_{E,0}}$ if $E\notin I$, where $\ket{\xi_{E,1}}$ and $\ket{\xi_{E,0}}$ are some states.

    We write down the algorithm.
    \begin{enumerate}
        \item Let $L=4n^c\beta$. 
        \item For $k\in \{0,1,\dots,L-1\}$:
        \begin{enumerate}
            \item For $\ell\in\{0,1,\dots,L\}$(If $k=0$, then $\ell\in\{1,2,\dots,L\}$):
            \begin{enumerate}
                \item Define the characteristic energy $\mathcal{E}_{k,\ell}:=\frac{\ell-1}{L}+\frac{k}{L^2}$, the interval $I_{k,\ell}:=[\mathcal{E}_{k,l},\mathcal{E}_{k,\ell}+\frac{1}{L})$, and $I'_{k,\ell}:=[\mathcal{E}_{k,\ell}-1/2L^2, \mathcal{E}_{k,\ell}+1/L+1/2L^2 )$.
                \item Define $U_{k,\ell}:=U_{\operatorname{dec}}(I'_{k,\ell})U_{\operatorname{EE}}(H,\delta E=\frac{1}{2L^2}, \eta=e^{-n}, \varepsilon=2^{-n})$, where $U_{\operatorname{EE}}$ is defined in \cref{def:energy_estimation}.
                \item Prepare the EPR state $\ket{\Psi}:=\sum_{k=1}^{N}\sqrt{\frac{1}{N}}\ket{k}\ket{k}$.
                \item Execute quantum counting (\cref{lem:amplitude_estimation}) on $U_{k,\ell}\ket{\Psi}\ket{0^{a_{U}}}$ up to a relative error $\dcoun=\frac{1}{4n^c}$, where $U_{k,\ell}$ acts on the first half of the EPR state and $a_{U}$ is the number of ancilla qubits.
                Let $\widetilde{M}_{k,\ell}$ be the output of the quantum counting.
            \end{enumerate}
            \item Let $\widetilde{Z}_{k}:=\sum_{\ell=0}^{L}\widetilde{M}_{k,\ell}e^{-\beta \mathcal{E}_{k,\ell}}$.
        \end{enumerate}
        \item Let $\widetilde{Z}:=\min\{\widetilde{Z}_k\}_{k=0}^{L-1}$.
        \item Output $\widetilde{Z}$.
    \end{enumerate}

    Now we claim the correctness of the algorithm.
    We define the following notations.
    Let $I^{-}_{k,\ell}:=[\mathcal{E}_{k,\ell}-1/L^2, \mathcal{E}_{k,\ell})$, $I^{+}_{k,\ell}:=[\mathcal{E}_{k,\ell}+1/L, \mathcal{E}_{k,\ell}+1/L+1/L^2)$, and $I''_{k,\ell}:=I^{-}_{k,\ell}\cup I_{k,\ell}\cup I^{+}_{k,\ell}$.
    Let $Z_{k,\ell}:=\sum_{j:E_j\in I_{k,\ell}} e^{-\beta E_{j}}$ be the partition function that are contributed by the states whose energies are in the interval $I_{k,\ell}$.
    Also, we define $Z^{-}_{k,\ell}:=\sum_{j:E_j\in I^{-}_{k,\ell}}e^{-\beta E_j}$ and $Z^{+}_{k,\ell}:=\sum_{j:E_j\in I^{+}_{k,\ell}}e^{-\beta E_j}$.
    Let $M_{k,\ell}$ be the number of states whose energies are in the interval $I_{k,\ell}$,
    $M^{-}_{k,\ell}$ be the number of states whose energies are in the interval $I^{-}_{k,\ell}$,
    and $M^{+}_{k,\ell}$ be the number of states whose energies are in the interval $I^{+}_{k,\ell}$.

    The quantum partition function is dominated by low energy. 
    The characteristic energy $\mathcal{E}_{k,\ell}$ is lower than all the energy in $I_{k, \ell}$. 
    Hence,  $M_{k,\ell}e^{-\mathcal{E}_{k,\ell}}$ gives a upper bound for $Z_{k, \ell}$, i.e.,
    \begin{equation}
    \label{eq:Z1}
        Z_{k,\ell} \le M_{k,\ell} e^{-\mathcal{E}_{k,\ell}}
    \end{equation}
    for any $k,\ell$.
    
    On the other hand, $\mathcal{E}_{k.\ell}$ is greater than all the energy decreased by $1/L$ in $I_{k, \ell}$.
    Hence, $M_{k,\ell}e^{-\mathcal{E}_{k,\ell}}$ gives a lower bound for $Z_{k, \ell}e^{1/L}$, i.e.,
    \begin{equation}
    \label{eq:Z2}
        M_{k,\ell} e^{-\mathcal{E}_{k,\ell}} \le Z_{k,\ell} e^{1/L}.
    \end{equation}
    for any $k,\ell$.
     
    Similarly, we have
    \begin{equation}
    \label{eq:Z3}
        M^{-}_{k,\ell} e^{-\beta\mathcal{E}_{k,\ell}} \le Z^{-}_{k,\ell},
    \end{equation}
    and
    \begin{equation}
    \label{eq:Z4}
        M^{+}_{k,\ell} e^{-\beta\mathcal{E}_{k,\ell}} \le Z^{+}_{k,\ell} e^{1/L+1/L^2}.
    \end{equation}

    First we consider the ideal case. That is, using $\widehat{U}_{\operatorname{EE}}(H, \delta E, \eta=0)$ (Equation \Cref{eq:energy_estimation}) instead of $U_{\operatorname{EE}}$ in Step 2.1.2.
    According to \cref{def:energy_estimation}, if $E_j\in I_{k,l}$, it holds that 
    \begin{equation*}
        \widehat{U}_{\operatorname{EE}}(H, \delta E, \eta=0)\ket{\psi_j}\ket{0^{a_{\operatorname{EE}}}} = \ket{\psi_j}\otimes\sum_{E'\in I'_{k,\ell}}\alpha_{E'}\ket{E'}\ket{g_{E'}};
    \end{equation*}
    and if  $E_j\notin I''_{k,\ell}$, it holds that 
    \begin{equation*}
        \widehat{U}_{\operatorname{EE}}(H, \delta E, \eta=0)\ket{\psi_j}\ket{0^{a_{\operatorname{EE}}}} = \ket{\psi_j}\otimes\sum_{E'\notin I'_{k,\ell}}\alpha_{E'}\ket{E'}\ket{g_{E'}}.
    \end{equation*}
    In other words, letting $\widehat{U}_{k,\ell}:=U_{\operatorname{dec}}(I'_{k,\ell})\widehat{U}_{\operatorname{EE}}(H, \delta E, \eta=0)$, we have
    \begin{equation}\label{eq:Ukl_result}
        \widehat{U}_{k,\ell}\ket{\psi_j}\ket{0^{a_{U}}}=\left\{
        \begin{array}{ll}
             \ket{1}\ket{\xi^1_{j}} &,\;\operatorname{if} E_j\in I_{k,\ell}, \\
             \alpha_{j}\ket{1}\ket{\xi^1_{j}}+\beta_{j}\ket{0}\ket{\xi^0_{j}}&,\; 
             \operatorname{if} E_j\in I^{-}_{k,\ell}\cup I^{+}_{k,\ell},\\
             \ket{0}\ket{\xi^0_{\ell,j}} 
             &,\;\operatorname{if} E_j\notin I''_{\ell},
        \end{array}
        \right.
    \end{equation}
    where $\alpha_{j},\beta_{j}\in\mathbb{C}$ satisfy  $|\alpha_{\ell,j}|+|\beta_{\ell,j}|=1$, $\ket{\xi^1_{j}}$, $\ket{\xi^0_{j}}$ are quantum states depending on $j$. The coefficients $\alpha_{j},\beta_{j}$ and the states $\ket{\xi^1_{j}}, \ket{\xi^0_{j}}$ may depend on $k,\ell$.

    When $\widehat{U}_{k,\ell}$ acts on the first half of $\ket{\Psi}$,
    combining Equation~(\ref{eq:Ukl_result}) and Lemma~\ref{lem:choi}, 
    we have
    \begin{equation}
        \widehat{U}_{k,\ell}\ket{\Psi}\ket{0^{a_U}} = \sqrt{\frac{\widehat{M}_{k,\ell}}{N}}\ket{1}\ket{\xi^1} + \sqrt{\frac{N-\widehat{M}_{k,\ell}}{N}}\ket{0}\ket{\xi^0}, 
    \end{equation}
    where $\widehat{M}_{k,\ell}\in[N]$ satisfies
    \begin{equation}
    \label{eq:M_ideal_result}
        M_{k,\ell} \le \widehat{M}_{k,\ell} \le M^{-}_{k,\ell}+M_{k,\ell}+M^{+}_{k,\ell},
    \end{equation}
    and $\ket{\xi^1}$, $\ket{\xi^0}$ are some quantum states.
    
    Let $\widetilde{M}^{ideal}_{k,\ell}$ be the output in Step 2.1.4.\ when replacing $U_{k,\ell}$ with $\widehat{U}_{k,\ell}$.  By the guarantee of quantum counting (\cref{lem:amplitude_estimation}), we have that
    \begin{equation}\label{eq:count_error}
        (1-\dcoun)\widehat{M}_{k,\ell}\le\widetilde{M}^{ideal}_{k,\ell}\le (1+\dcoun)(\widehat{M}_{k,\ell}).
    \end{equation}
    By union bound, the event that Equation \cref{eq:count_error} holds for all $k, \ell$ has probability $1-\negl(n)$.

    Now we lower bound $\widetilde{Z}$.
    Because $M_{k,\ell} \le \widehat{M}_{k,\ell}$, from Equation \cref{eq:count_error}, we have
    \begin{equation*}
        (1-\dcoun)M_{k,\ell} \le \widetilde{M}^{ideal}_{k,\ell}.    
    \end{equation*}
    Multiplying by $e^{-\beta \mathcal{E}_{k,\ell}}$ and combining with Equation \cref{eq:Z1}, we have  
    \begin{equation}
    \label{eq:lower_1}
        (1-\dcoun)Z_{k,\ell}  \le \widetilde{M}^{ideal}_{k,\ell}e^{-\beta \mathcal{E}_{k,\ell}}.
    \end{equation}
    For any $k$, when summing $\ell$ over $\{0,1,\dots, L\}$ in Equation \cref{eq:lower_1}, we have 
    \begin{equation}
    \label{eq:lower_2}
        (1-\dcoun)Z \le \sum_{\ell}\widetilde{M}^{ideal}_{k,\ell}e^{-\beta \mathcal{E}_{k,\ell}}.
    \end{equation}
    
    For the upper bound, because $\widehat{M}_{k,\ell}\le M^{-}_{k,\ell}+M_{k,\ell}+M^{+}_{k,\ell}$, combining with the guarantee of quantum counting (Equation \cref{eq:count_error}),
    we have
    \begin{equation*}
        \widetilde{M}^{ideal}_{k,\ell} \le (1+\dcoun)(M^{-}_{k,\ell}+M_{k,\ell}+M^{+}_{k,\ell}). 
    \end{equation*}
    Multiplying by $e^{-\beta\mathcal{E}_{k,\ell}}$ and combing with Equation \cref{eq:Z3} and \cref{eq:Z4}, we have
    \begin{equation}
    \label{eq:upper_1}
        \widetilde{M}^{ideal}_{k,\ell}e^{-\beta \mathcal{E}_{k,\ell}} \le (1+\dcoun)(Z^{-}_{k,\ell}+Z_{k,\ell}e^{-\beta/L}+Z^{+}_{k,\ell}e^{\beta/L+\beta/L^2}). 
    \end{equation}
    Fixing a $k$ and summing $\ell$ over $\{0,1,\dots, L\}$ in Equation \cref{eq:upper_1}, we have 
    \begin{equation}
    \label{eq:upper_2}
        \sum_{\ell}\widetilde{M}^{ideal}_{k,\ell}e^{-\beta \mathcal{E}_{k,\ell}} \le (1+\dcoun)\Big(\big(\sum_{\ell}Z^{-}_{k,\ell}\big)+Ze^{\beta/L} +\big(\sum_{\ell}Z^{+}e^{\beta/L+\beta/L^2}\big)\Big).
    \end{equation}
    Consider the term $\sum_{\ell}Z^{-}_{k,\ell}+\sum_{\ell}Z^{+}e^{\beta/L+\beta/L^2}$ and sum over $k$. That is
    \begin{equation}
        \label{eq:upper_3}
        \sum_{k=0}^{L-1}\Big(\sum_{\ell}Z^{-}_{k,\ell}+\sum_{\ell}Z^{+}e^{\beta/L+\beta/L^2}\Big). 
    \end{equation}
    Since
    \begin{equation*}
        \bigcup_{k=0}^{L-1}\bigcup_{\ell}I^{-}_{k,\ell} =[-1/L,0) \cup [0,1-1/L^2),
    \end{equation*}
    and all the eigenstates are in the interval $[0,1)$,
    we have that $\sum_{k=0}^{L-1}\sum_{\ell}Z^{-}_{k,\ell}$ is contributed by the eigenstates in the interval $[0,1-1/L^2)$.
    Therefore,
    \begin{equation}
        \label{eq:Zminus_upper}
        \sum_{k=0}^{L-1}\sum_{\ell}Z^{-}_{k,\ell} \le Z.
    \end{equation} 
    Similarly, we have
    \begin{equation}
        \label{eq:Zplus_upper}
        \sum_{k=0}^{L-1}\sum_{\ell}Z^{+}_{k,\ell} \le Z.
    \end{equation} 
    Plugging Equation \cref{eq:Zminus_upper} and \cref{eq:Zplus_upper} into Equation \cref{eq:upper_3}, we have
    \begin{equation}
        \label{eq:upper4}
        \sum_{k=0}^{L-1}\Big(\sum_{\ell}Z^{-}_{k,\ell}+\sum_{\ell}Z^{+}e^{\beta/L+\beta/L^2}\Big)\le Z(1+e^{\beta/L+\beta/L^2}). 
    \end{equation}
    By average argument, there is a $k$ such that 
    \begin{equation}
        \label{eq:average_arg}
        \Big(\sum_{\ell}Z^{-}_{k,\ell}+\sum_{\ell}Z^{+}e^{\beta/L+\beta/L^2}\Big)\le \frac{1}{L}Z(1+e^{\beta/L+\beta/L^2}). 
    \end{equation}
    Let $\widetilde{Z}^{ideal}$ be the output of the algorithm when replacing $U_{k,\ell}$ with $\widetilde{U}_{k,\ell}$.
    By Equation \cref{eq:upper_2} and \cref{eq:average_arg}, we have
    \begin{equation}
        \label{eq:upper_5}
        \widetilde{Z}^{ideal}\le (1+\dcoun)\big(\frac{1}{L}+e^{\beta/L}+\frac{1}{L}e^{\beta/L+\beta/L^2}\big)Z.
    \end{equation}
    Recall that $\dcoun=\frac{1}{4n^c}$,  $L=4n^c\beta$, and $\beta=\poly(n)$.
    We have 
    \begin{equation*}
        \frac{1}{L}+e^{\beta/L}+\frac{1}{L}e^{\beta/L+\beta/L^2}
        =e^{\frac{1}{4n^c}}+\frac{e^{\frac{1}{4n^c}(1+\frac{1}{4n^c\beta})}}{4n^c\beta} +\frac{1}{4n^c\beta}\le 1+\frac{1}{2n^c}
    \end{equation*}
    for sufficiently large $n$.
    As a result, Equation \cref{eq:upper_5} becomes
    \begin{equation}
        \label{eq:upper_6}
        \widetilde{Z}^{ideal}\le \big(1+\frac{1}{4n^c}\big)\big(1+\frac{1}{2n^c}\big)Z\le \big(1+\frac{1}{n^c}\big)Z
    \end{equation}
    for sufficiently large $n$.

    Finally, combing Equation \cref{eq:lower_2} and \cref{eq:upper_6}, we have
    \begin{equation}
        \label{eq:Z_ideal_redult}
        (1-1/n^c)Z\le \widetilde{Z}^{ideal}\le (1+1/n^c)Z.
    \end{equation}
    
    Then, we consider the real implementation.
    There are two steps to change the ideal algorithm to the real implementation.
    First, we define a intermediate algorithm by replacing $\widehat{U}_{\operatorname{EE}}(H,\delta E, \eta=0)$ in the ideal algorithm 
    with $\widehat{U}_{\operatorname{EE}}(H,\delta E, \eta=e^{-n})$. 
    By \cref{lem:amplitude_estimation}, there are $O(\poly(n)\cdot 2^{\frac{n}{2}})$ many $\widehat{U}_{EE}$ in the algorithm.
    Also, we have $\|\widehat{U}_{\operatorname{EE}}(H,\delta E, \eta=e^{-n})-\widehat{U}_{\operatorname{EE}}(H,\delta E, \eta=0)\|\le e^{-n}$. 
    By triangular inequity, the distance between the ideal algorithm  and  the intermediate algorithm is
    $\negl(n)$.

    For the second step, we get the real implementation by replacing $\widehat{U}_{\operatorname{EE}}(H,\delta E, \eta=e^{-n})$ with ${U}_{\operatorname{EE}}(H,\delta E, \eta=e^{-n}, \varepsilon=2^{-n})$.
    We have $\|U_{\operatorname{EE}}(H,\delta E, \eta=e^{-n}, \varepsilon=2^{-n})-\widehat{U}_{\operatorname{EE}}(H,\delta E, \eta=e^{-n})\|\le 2^{-n}$. 
    Again, by triangular inequity, the distance between the intermediate algorithm and the real implementation is 
    $\negl(n)$.
    As a result, the distance between ideal algorithm and the real implementation is 
    $\negl(n)$.
    Constituently, the real implementation outputs the correct estimation of quantum partition function with probability $1-\negl(n)$.
    
    Finally, we analyze the running time. There are $\poly(n)$ of iterations over $k$ and $\ell$ and each iteration runs in $O(\poly(n)\cdot 2^{\frac{n}{2}})$ time.
    Therefore the overall running time is $O^{\ast}(2^{\frac{n}{2}})$.
\end{proof}

\section*{Acknowledgments}

NHC and YCS thank Chunhao Wang and Christopher Ye for valuable discussions. FLG thanks Suguru Tamaki for discussions. The authors thank anonymous referees for their comments.

NHC is supported by NSF Award FET-2535028, NSF Award FET-2243659, NSF Career Award FET-2339116, Google Scholar Award, and DOE Quantum Testbed Finder Award DE-SC0024301. YCS is supported by NSF Award FET-2339116 and NSF Award FET-2243659. AH is supported by JSPS KAKENHI grant No.~24H00071, 25K24674, 25K2446. FLG is supported by JSPS KAKENHI grant No.~24H00071, 25K24674, 25K2446, MEXT Q-LEAP grant No.~JPMXS0120319794, JST ASPIRE grant No.~JPMJAP2302 and JST CREST grant No.~JPMJCR24I4.

\bibliographystyle{alpha}
\bibliography{ref}

\appendix

\section{A trivial fine-grained reduction to $k(\varepsilon)$LH from $k(\varepsilon)$SAT}
\label{sec:appexdixA}

\begin{theorem}[Lower bound of $k(\varepsilon)$LH]
    \label{thm:almost_make_me_cry_thm}
    Assuming SETH (resp. QSETH), for any $\varepsilon > 0$, there is $k$ (depending on $\varepsilon$) such that for any algorithm, there exists $n_0\in\mathbb{N}$ such that for all $n\ge n_0$, there is a Hamiltonian $H$ acting on $n$ qubits, associated $\EnergyYes,\EnergyNo$ satisfying $\EnergyNo-\EnergyYes\ge\Omega(1)$ such that $\LH(H, \EnergyYes, \EnergyNo)$ cannot be solved in $O(2^{n(1-\varepsilon)})$ classical time (resp. $O(2^{\frac{n}{2}(1-\varepsilon)})$ quantum time).
\end{theorem}
\begin{proof}
    We construct a Hamiltonian $H$ from a $\kSAT$ instance $\Phi=\varphi_1\wedge \varphi_2\wedge\cdots\wedge\varphi_m$.

    Let $S_i\subseteq[n]$ be the set collecting the index $j$ such that $x_j$ or $\neg x_j$ appears in $\varphi_i$, and let $\rgst{S_i}$ be the corresponding register. 
    Let $y_i\in\{0,1\}^{|S_i|}$ be the assignment to the variables appearing in $\varphi_i$ such that $\varphi_i(y_i)=0$.
    Because $\varphi_i$ is in the disjunctive form, $y_i$ is unique.
    Let $H:=\sum_{i=1}^{m}H_i$, where $H_i:= I-\proj{y_i}\reg{S_i}$ for all $i\in[m]$.
    It holds that $H_i\ket{y_i}=0$ if $\varphi_i(y_i)=1$ and $H_i\ket{y_i}=\ket{y_i}$ if $\varphi_i(y_i)=0$.
    Each $H_i$ acts non-trivially on at most $k$ qubits.
    Hence, $H$ is $k$-local.
    
    If there exists an assignment $x\in\{0,1\}^n$ satisfying $\Phi$, then $H\ket{x}=0$. Otherwise, $\lambda(H)\ge 1$.
    Hence, solving $\LH(H, \EnergyYes=0,\EnergyNo=1)$ can decide $\kSAT$.
    The construction of $H$ takes $\poly(n)$ time.
    Hence, if $\LH(H, \EnergyYes=0,\EnergyNo=1)$ is solved in $O(2^{n}(1-\varepsilon))$ time, then $\kSAT(\Phi)$ is also solved in $O(2^{n(1-\varepsilon)})$ classical time, which violates SETH. The argument assuming QSETH follows in the same proof strategy.
\end{proof}

\end{document}